\documentclass[11pt,onecolumn]{IEEEtran}
\usepackage{amsthm}
\usepackage{times,amssymb,amsmath,amsfonts,float,nicefrac,color,bbm,mathrsfs,float}
\usepackage[dvipsnames]{xcolor}
\usepackage{algorithm,enumerate,multirow,caption,tikz,graphicx}
\usepackage{algpseudocode}
\usepackage[noadjust]{cite}
\usepackage{subcaption}
\usepackage{comment}

\usepackage[top=1in,bottom=1in,left=1in,right=1in]{geometry}
\allowdisplaybreaks

\newcommand{\RN}[1]{%
  \textup{\expandafter{\romannumeral#1}}%
}

\tikzset{
  block/.style    = {draw, thick, rectangle, minimum width = 3em},
  sblock/.style      = {draw, thick, rectangle, minimum height = 3em,
    minimum width = 3em}, 
}

\newcommand\remove[1]{}

\allowdisplaybreaks

\newtheorem{definition}{Definition}
\newtheorem{proposition}{Proposition}

\newtheorem{lemma}{Lemma}

\newtheorem{remark}{Remark}

\newtheorem{cnstr}{Construction}

\def\mathbi#1{{\textbf{\textit #1}}}

\newcommand{\bB}{\mathbb{B}}

\newcommand{\bE}{\mathbb{E}}

\newcommand{\bP}{\mathbb{P}}

\newcommand{\cA}{\mathcal{A}}

\newcommand{\cC}{\mathcal{C}}

\newcommand{\cM}{\mathcal{M}}

\newcommand{\cQ}{\mathcal{Q}}
\newcommand{\cR}{\mathcal{R}}

\newcommand{\cW}{\mathcal{W}}

\DeclareMathOperator{\Est}{Est}

\DeclareMathOperator{\sign}{sign}
\DeclareMathOperator{\spn}{span}

\DeclareMathOperator{\LLR}{LLR}

\DeclareMathOperator{\Proj}{Proj}

\DeclareMathOperator{\argmax}{argmax}

\begin{document}
\title{Recursive projection-aggregation decoding of Reed-Muller codes}

\author{Min Ye \and \hspace*{.4in} Emmanuel Abbe}

\maketitle
{\renewcommand{\thefootnote}{}\footnotetext{

\vspace{-.2in}
 
\noindent\rule{1.5in}{.4pt}

A preliminary version of this paper was presented at the 2019 IEEE International Symposium on Information Theory, July 2019, Paris, France \cite{Ye19}.

M. Ye is with the Data Science and Information Technology Research Center, Tsinghua-Berkeley Shenzhen Institute, Shenzhen, China. Email: yeemmi@gmail.com

E. Abbe is with the Mathematics Institute and the School of Computer and Communication Sciences at EPFL, Switzerland, and the Program in Applied and Computational Mathematics and the Department of Electrical Engineering in Princeton University, USA. 
}

\renewcommand{\thefootnote}{\arabic{footnote}}
\setcounter{footnote}{0}

\begin{abstract}
We propose a new class of efficient decoding algorithms for Reed-Muller (RM) codes over binary-input memoryless channels. 
The algorithms are based on  projecting the code on its cosets, recursively decoding the projected codes (which are lower-order RM codes), and aggregating the reconstructions (e.g., using majority votes). We further provide extensions of the algorithms using list-decoding. 

We run our algorithm for AWGN channels and Binary Symmetric Channels at the short code length ($\le 1024$) regime for a wide range of code rates. Simulation results show that in both low code rate and high code rate regimes, the new algorithm  outperforms the widely used decoder for polar codes (SCL+CRC) with the same parameters. The performance of the new algorithm for RM codes in those  regimes is in fact close to that of the maximal likelihood decoder. Finally, the new decoder naturally allows for parallel implementations.
\end{abstract}

\section{Introduction}\label{sect:intro}

Reed-Muller (RM) codes are among the oldest families of error-correcting codes \cite{Reed54}. 
The recent breakthrough of polar codes \cite{Arikan09} has brought the attention back to RM codes, due to the closeness of the two codes. RM codes  have  in particular the advantage of having a simple and universal code construction, and promising performances were demonstrated in several works \cite{Arikan08,Mondelli14}, with a scaling law conjectured to be comparable of that of random codes.

RM codes do not possess yet the  generic analytical framework of polar codes (i.e., polarization theory). It was recently shown that RM codes achieve capacity on the Binary Erasure Channel (BEC) at constant rate \cite{Kudekar17}, as well as for extremal rates for BEC and Binary Symmetric Channels (BSC)  \cite{Abbe15}, but obtaining such results for a broader class of communication channels and rates  remains  open. Recent progress was made on these questions with a polarization approach to RM codes shown in \cite{AY18}.
See also \cite{ASY20} for a recent survey on RM codes.

Various decoding algorithms have been proposed for  RM codes, starting with Reed algorithm \cite{Reed54,Macwilliams77}, and four important more recent line of works including automorphism group based decoding \cite{Sidel92,Loidreau04,Sakkour05}, recursive list-decoding \cite{Dumer04,Dumer06,Dumer06a}, a new Berlekamp-Welch type of algorithm \cite{Saptharishi17,Sberlo18}, and a new algorithm utilizing minimum-weight parity checks \cite{Santi18}.
In particular, \cite{Sidel92,Dumer04,Dumer06,Dumer06a,Saptharishi17,Sberlo18} give fairly powerful theoretical guarantees for efficient decoding of RM codes in specific regimes.
However, there is not a thorough comparison between the performance of RM codes under these decoders and the performance of the widely used CRC-aided polar codes under the Successive Cancellation List (SCL) decoders \cite{Tal15}.

In this paper, we propose a new class of decoding algorithms for Reed-Muller codes over any binary-input memoryless channels and compare its performance with polar codes. 
The new algorithms are based on recursive projections and aggregations of cosets decoding, exploiting the self-similarity of RM codes, and are extended  with Chase  list-decoding algorithms \cite{Chase72}. 
We run our new algorithms at the short code length ($\le 1024$) regime for a wide range of code rates.
Simulation results show that
the new algorithms improve on the widely used decoding algorithm for polar codes \cite{Tal15} in both low code rate and high code rate regimes.
These are the type of regimes where polar codes are planned to enter the 5G standards \cite{3gpp} as well as relevant regimes for applications in the Internet of Things (IoT).

More specifically, we compare our new algorithm for RM codes with the Successive Cancellation List (SCL) decoder for CRC-aided polar codes \cite{Tal15}, where we set the CRC size to take  optimal  values\footnote{The optimal CRC size depends on the choice of code length and rate.}. For AWGN channels, our new algorithm has about $0.5$dB gain (more in some cases) over polar codes in  various short code length ($\le 1024$) and low code rate ($\le 0.5$) regimes, and similar improvements are also obtained for BSC channels.
Moreover, the performance of our new decoding algorithm is comparable to the best previously known algorithms for RM codes \cite{Dumer06a}.

In the above regimes, the decoding error probability of  our new algorithm is in fact  shown to  be close to that of the Maximal Likelihood decoder on  RM codes. Some extensions and variants to potentially further improve the performance are also discussed, as well as possible extensions of the projection-aggregation algorithms to other families of codes.

In Section~\ref{sect:hlvl}, we give a high level description of the new type of algorithms. In Section~\ref{sect:BSC}, we present decoding algorithm for BSC channels. In Section~\ref{sect:gen} we generalize the algorithms to decode RM codes over any binary-input channel. Finally, in Section~\ref{sect:sim} we present simulation results. In addition to the previously mentioned improvements over polar codes, we also empirically validate the improved scaling-law of RM codes over polar codes on BSC channels \cite{Hassani18}.

\section{A high-level description of the new algorithms}  \label{sect:hlvl}

We begin with some notation and background on RM codes. In this paper, we use $\oplus$ to denote sums over $\mathbb{F}_2$.
Let us consider the polynomial ring $\mathbb{F}_2[Z_1,Z_2,\dots,Z_m]$ of $m$ variables. Since $Z^2=Z$ in $\mathbb{F}_2$, the following set of $2^m$ monomials forms a basis of $\mathbb{F}_2[Z_1,Z_2,\dots,Z_m]$:
$$
\{\prod_{i\in A}Z_i: A\subseteq [m] \}, \text{~where~} \prod_{i\in \emptyset}Z_i :=1.
$$
Next we associate every subset $A\subseteq [m]$ with a row vector $\mathbi{v}_m(A)$ of length $2^m$, whose components are indexed by a binary vector $\mathbi{z}=(z_1,z_2,\dots,z_m) \in \{0,1\}^m$.
The vector $\mathbi{v}_m(A)$ is defined as follows:
\begin{equation}\label{eq:gg}
\mathbi{v}_m(A,\mathbi{z}) = \prod_{i\in A} z_i,
\end{equation}
where $\mathbi{v}_m(A,\mathbi{z})$ is the component of $\mathbi{v}_m(A)$ indexed by $\mathbi{z}$,
i.e., $\mathbi{v}_m(A,\mathbi{z})$ is the evaluation of the monomial $\prod_{i\in A}Z_i$ at $\mathbi{z}$.
For $0\le r \le m$, the set of vectors 
$$
\{\mathbi{v}_m(A):A\subseteq[m],|A|\le r\}
$$
forms a basis of the $r$-th order Reed-Muller code $\cR\cM(m,r)$ of length $n:=2^m$
and dimension $\sum_{i=0}^r \binom{m}{i}$.

\begin{definition} \label{def:kee}
 The $r$-th order Reed-Muller code $\cR\cM(m,r)$ code is defined as the following set of binary vectors 
$$
\cR\cM(m,r) := \left\{\sum_{A\subseteq[m],|A|\le r}u(A) \mathbi{v}_m(A): u(A)\in\{0,1\} 
\text{~~for all~} A\subseteq[m],|A|\le r\right\}.
$$
\end{definition}

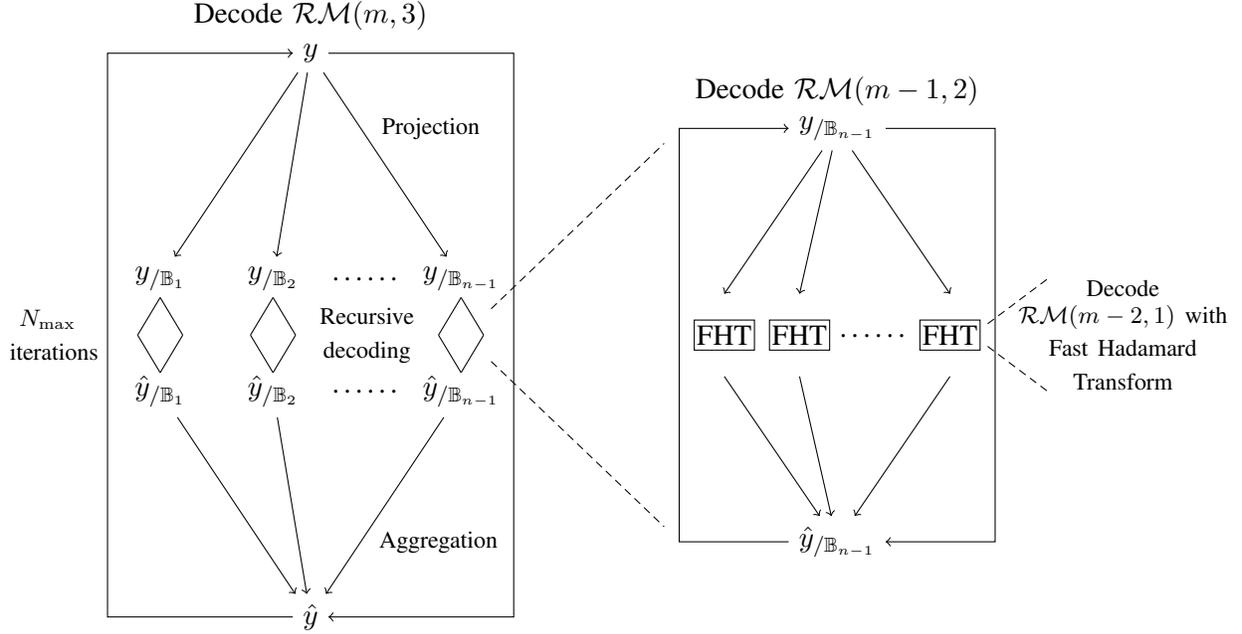
\begin{figure}
\centering
\begin{tikzpicture}
\node at (0, 9.5) [] {Decode $\cR\cM(m,3)$};
\node at (-3.5, 5.25) [text width=1cm,align=center]  {\footnotesize{$N_{\max}$ \\ iterations}} ;
\node at (0.75, 5.25) [text width=1.5cm,align=center] {\footnotesize{Recursive decoding}};
\node at (0.75, 6) [] {\dots\dots};
\node at (0.75, 4.5) [] {\dots\dots};
\node at (1.6, 8) [] {\footnotesize{Projection}};
\node at (1.7, 2.5) [] {\footnotesize{Aggregation}};
\node at (0,9) [] (y) {$y$} ;
\node at (-2,6) [] (yb1) {$y_{/\mathbb{B}_1}$};
\node at (-0.5,6) [] (yb2) {$y_{/\mathbb{B}_2}$};
\node at (2,6) [] (ybn) {$y_{/\mathbb{B}_{n-1}}$};
\node at (-2,4.5) [] (hyb1) {$\hat{y}_{/\mathbb{B}_1}$};
\node at (-0.5,4.5) [] (hyb2) {$\hat{y}_{/\mathbb{B}_2}$};
\node at (2,4.5) [] (hybn) {$\hat{y}_{/\mathbb{B}_{n-1}}$};
\node at (0, 1.5) [] (hy) {$\hat{y}$};
\draw [->] (y)--(yb1);
\draw [->] (y)--(yb2);
\draw [->] (y)--(ybn);
\draw [->] (hyb1)-- (hy);
\draw [->] (hyb2)-- (hy);
\draw [->] (hybn)-- (hy);
\draw [->] (hy) -- ++(-2.7, 0) -- ++(0, 7.5) -- (y);
\draw [->] (y) -- ++(2.7, 0) -- ++(0, -7.5) -- (hy);

\draw [-] (-2, 5.75)--(-2.3, 5.25);
\draw [-] (-2.3, 5.25)--(-2, 4.75);
\draw [-] (-2, 5.75)--(-1.7, 5.25);
\draw [-] (-1.7, 5.25)--(-2, 4.75);

\draw [-] (-0.5, 5.75)--(-0.8, 5.25);
\draw [-] (-0.8, 5.25)--(-0.5, 4.75);
\draw [-] (-0.5, 5.75)--(-0.2, 5.25);
\draw [-] (-0.2, 5.25)--(-0.5, 4.75);

\draw [-] (2, 5.75)--(1.7, 5.25);
\draw [-] (1.7, 5.25)--(2, 4.75);
\draw [-] (2, 5.75)--(2.3, 5.25);
\draw [-] (2.3, 5.25)--(2, 4.75);

\node at (7, 8.5) [] {Decode $\cR\cM(m-1,2)$};
\node at (7.5, 5.25) [] {\dots\dots};

\node at (7,8) [] (ys) {$y_{/\mathbb{B}_{n-1}}$} ;
\node at (7,2.5) [] (hys) {$\hat{y}_{/\mathbb{B}_{n-1}}$} ;
\draw [->] (ys)--(5.5,5.8);
\draw [->] (ys)--(6.5,5.8);
\draw [->] (ys)--(8.5,5.8);

\draw [->] (5.5, 4.7)--(hys);
\draw [->] (6.5, 4.7)--(hys);
\draw [->] (8.5, 4.7)--(hys);

\node at (5.5, 5.25) [] {FHT};
\node at (6.5, 5.25) [] {FHT};
\node at (8.5, 5.25) [] {FHT};

\draw [] (5.1, 5.05) rectangle ++ (0.8, 0.4);
\draw [] (6.1, 5.05) rectangle ++ (0.8, 0.4);
\draw [] (8.1, 5.05) rectangle ++ (0.8, 0.4);

\draw [->] (hys) -- ++(-2.1, 0) -- ++(0, 5.5) -- (ys);
\draw [->] (ys) -- ++(2.1, 0) -- ++(0, -5.5) -- (hys);

\draw [densely dashed] (2.4, 5.6) -- (4.7, 7.8);
\draw [densely dashed] (2.4, 4.9) -- (4.7, 2.7);

\draw [densely dashed] (9, 5.4) -- (9.8, 6);
\draw [densely dashed] (9, 5.1) -- (9.8, 4.5);

\node at (10.8, 5.25) [text width=3cm,align=center] 
{\footnotesize{Decode $\cR\cM(m-2,1)$ with\\ Fast Hadamard Transform}};

\end{tikzpicture}
\caption{Recursive Projection-Aggregation decoding algorithm for third order RM codes} \label{fig:highlvl}
\end{figure}

\begin{algorithm}
\caption{The \texttt{RPA\_RM} decoding function for BSC}    \label{alg:highlvl}
\textbf{Input:} The corrupted codeword $y=(y(\mathbi{z}), \mathbi{z}\in\mathbb{E})$; the parameters of the Reed-Muller code $m$ and $r$; the maximal number of iterations $N_{\max}$

\textbf{Output:} The decoded codeword $\hat{c}$

\vspace*{0.05in}
\begin{algorithmic}[1]
\For {$j=1,2,\dots,N_{\max}$} 
\State $y_{/\mathbb{B}_i} \gets \Proj(y,\mathbb{B}_i)$ for $i=1,2,\dots,2^m-1$     \Comment{Projection}
\State $\hat{y}_{/\mathbb{B}_i} \gets \texttt{RPA\_RM}(y_{/\mathbb{B}_i},m-1,r-1,N_{\max})$ for $i=1,2,\dots,2^m-1$ 
\Comment{Recursive decoding}
\State \Comment{If $r=2$, then we use the Fast Hadamard Transform to decode the first-order RM code \cite{Macwilliams77}}
\State $\hat{y} \gets \texttt{Aggregation}(y,\hat{y}_{/\mathbb{B}_1},\hat{y}_{/\mathbb{B}_2}\dots,\hat{y}_{/\mathbb{B}_{n-1}})$           \Comment{Aggregation}
\If {$y=\hat{y}$}
\State \textbf{break}  \Comment{$y=\hat{y}$ means that the algorithm already converges to a fixed (stable) point}
\EndIf
\State $y \gets \hat{y}$
\EndFor
\State $\hat{c} \gets \hat{y}$
\State \textbf{return} $\hat{c}$
\end{algorithmic}
\end{algorithm}

In other words, each vector $\mathbi{v}_m(A)$ consists of all the evaluations of the monomial $\prod_{i\in A}Z_i$ at all the points in the vector space $\mathbb{E}:=\mathbb{F}_2^m$, and each codeword $c\in\cR\cM(m,r)$ corresponds to an $m$-variate polynomial with degree at most $r$.
The coordinates of the codeword $c$ are also indexed by the binary vectors $\mathbi{z}\in\mathbb{E}$, and we write $c=(c(\mathbi{z}), \mathbi{z}\in\mathbb{E})$.
Let $\mathbb{B}$ be an $s$-dimensional subspace of $\mathbb{E}$, where $s\le r$.
The quotient space $\mathbb{E}/\mathbb{B}$ consists of all the cosets of $\mathbb{B}$ in $\mathbb{E}$, where every coset $T$ has form $T=\mathbi{z}+\mathbb{B}$ for some $\mathbi{z}\in\mathbb{E}$.
For a binary vector $y=(y(\mathbi{z}), \mathbi{z}\in\mathbb{E})$, we define its projection on the cosets of $\mathbb{B}$ as
\begin{equation}  \label{eq:qb}
y_{/\mathbb{B}} =\Proj(y,\mathbb{B}) := \Big( y_{/\mathbb{B}}(T), T\in\mathbb{E}/\mathbb{B} \Big), \text{~where~} y_{/\mathbb{B}}(T) := \bigoplus_{\mathbi{z}\in T} y(\mathbi{z} )
\end{equation}
is the binary vector obtained by summing up all the coordinates of $y$ in each coset $T\in\mathbb{E}/\mathbb{B}$. Here the sum is over $\mathbb{F}_2$ and the dimension of $y_{/\mathbb{B}}$ is $n/|\mathbb{B}|$.

In the next section, we will show that if $c$ is a codeword of $\cR\cM(m,r)$, then $c_{/\mathbb{B}}$ is a codeword of $\cR\cM(m-s,r-s)$, where $s$ is the dimension of $\mathbb{B}$.
Our new decoding algorithm makes use of the case $s=1$, namely, the one-dimensional subspaces.
More precisely, let $y=(y(\mathbi{z}), \mathbi{z}\in\mathbb{E})$ be the output vector of transmitting a codeword of $\cR\cM(m,r)$ over some BSC channel.
Our decoding algorithm is defined in a recursive way: For every one-dimensional subspace $\mathbb{B}$, we first obtain the projection $y_{/\mathbb{B}}$, and then we use the decoding algorithm for $\cR\cM(m-1,r-1)$ to decode $y_{/\mathbb{B}}$, where the decoding result is denoted as $\hat{y}_{/\mathbb{B}}$.
Since every one-dimensional subspace of $\mathbb{E}$ consists of $0$ and a non-zero element, there are $n-1$ such subspaces in total.
After the projection and recursive decoding steps, we obtain $n-1$ decoding results $\hat{y}_{/\mathbb{B}_1}, \hat{y}_{/\mathbb{B}_2}, \dots, \hat{y}_{/\mathbb{B}_{n-1}}$.
Next we use a majority voting scheme to aggregate these decoding results together with $y$ to obtain a new estimate $\hat{y}$ of the original codeword.
Finally we update $y$ as $\hat{y}$, and iterate the whole procedure for up to $N_{\max}$ rounds.
Notice that if $y=\hat{y}$ (see line 6), then $y$ is a fixed (stable) point of this algorithm and will remain unchanged for the next iterations. In this case we should exit the for loop on line 1 (see line 6--8).
In practice we set the maximal number of iterations $N_{\max}=\lceil m/2 \rceil$ to prevent the program from running into an infinite loop, and typically $\lceil m/2 \rceil$ iterations are enough for the algorithm to converge to a stable $y$.
This high-level description is summarized in Fig.~\ref{fig:highlvl} and Algorithm~\ref{alg:highlvl}.
While this description focuses on the decoding algorithm over BSC, a natural extension of this algorithm bases on log-likelihood ratios (LLRs) allows us to decode RM codes over any binary-input memoryless channels, including the AWGN channel; see Section~\ref{sect:gen} for details.

\subsection{List decoding procedure \cite{Chase72}}   \label{sect:highlist}
Here we recap (a version of) the list decoding procedure proposed by Chase \cite{Chase72} that can further decrease the decoding error probability. Suppose that we have a unique decoding algorithm \texttt{decodeC} for some code $\cC$ over some binary-input memoryless channel $W:\{0,1\}\to\cW$. Without loss of generality, assume that \texttt{decodeC} is based on the LLR vector of the channel output, where the LLR of an output symbol $x\in\cW$ is defined as
\begin{equation}\label{eq:defllr}
\LLR(x):=\ln\Big(\frac{W(x|0)}{W(x|1)}\Big).
\end{equation}
Clearly, if $|\LLR(x)|$ is small, then $x$ is a noisy symbol, and if $|\LLR(x)|$ is large, then $x$ is relatively noiseless.

The list decoding procedure works as follows. 
Suppose that $y=(y_1,y_2,\dots,y_n)$ is the output vector when we send a codeword of $\cC$ over the channel $W$. We first sort $|\LLR(y_i)|,i\in[n]$ from small to large. Without loss of generality, let us assume that $|\LLR(y_1)|,|\LLR(y_2)|,|\LLR(y_3)|$ are the three smallest components in the LLR vector, meaning that $y_1,y_2$ and $y_3$ are the three most noisy symbols in the channel outputs (we take three arbitrarily). Next we enumerate all the possible cases of the first three bits of the codeword $c=(c_1,c_2,\dots,c_n)$: The first three bits $(c_1,c_2,c_3)$ can be any vector in $\mathbb{F}_2^3$, so there are 8 cases in total, and for each case we change the value of $\LLR(y_1),\LLR(y_2),\LLR(y_3)$ according to the values of $c_1,c_2,c_3$. More precisely, we set $\LLR(y_i)=(-1)^{c_i} L_{\max}$ for $i=1,2,3$, where $L_{\max}$ is some large real number. In practice, we can choose $L_{\max}:=\max(|\LLR(y_i)|,i\in[n])$ or $L_{\max}:=2\max(|\LLR(y_i)|,i\in[n])$.
For each of these 8 cases, we use \texttt{decodeC} to obtain a decoded codeword, and we denote them as $\hat{c}^{(1)},\hat{c}^{(2)},\dots,\hat{c}^{(8)}$.
Finally, we calculate the posterior probability of $W^n(y|\hat{c}^{(i)}), 1\le i\le 8$, and choose the largest one as the final decoding result, namely, we perform a maximal likelihood decoding among the 8 candidates in the list.

When we apply this list decoding procedure together with Algorithm~\ref{alg:highlvl} to decode RM codes, the decoding error probability is typically close to that of the Maximal Likelihood decoder.


\section{Decoding algorithm for BSC} \label{sect:BSC}

We begin with the definition of the quotient code. Then we show that the quotient code of an RM code is also an RM code.
\begin{definition}
Let $s\le r\le m$ be integers, and let $\mathbb{B}$ be an $s$-dimensional subspace of $\mathbb{E}:=\mathbb{F}_2^m$. We define the quotient code
$$
\cQ(m,r,\mathbb{B}):=
\{c_{/\mathbb{B}} : c\in\cR\cM(m,r)\}  .
$$
\end{definition}

\begin{lemma}\label{lm:tbs}
Let $s\le r\le m$ be integers, and let $\mathbb{B}$ be an $s$-dimensional subspace of $\mathbb{E}:=\mathbb{F}_2^m$. The code $\cQ(m,r,\mathbb{B})$ is the Reed-Muller code $\cR\cM(m-s,r-s)$.
\end{lemma}
This lemma is an immediate corollary of Theorem 12 in \cite[Chapter 13]{Macwilliams77}. For the sake of completeness, we give a proof of this lemma in Appendix~\ref{ap:tbs}.

Note that Reed's algorithm \cite{Reed54} relies on the special case of $s=r$ in Lemma~\ref{lm:tbs}, and our new decoding algorithm makes use of the case $s=1$ in Lemma~\ref{lm:tbs} (in addition to using all subspaces and adding an iterative process).
The \texttt{RPA\_RM} decoding function is already presented in the previous section. Here we fill in the only missing component, namely the \texttt{Aggregation} function; see Algorithm~\ref{alg:BSCAg} below.
Both $y_{/\mathbb{B}_i}=(y_{/\mathbb{B}_i}(T),T\in\mathbb{E}/\mathbb{B})$ and 
$\hat{y}_{/\mathbb{B}_i}=(\hat{y}_{/\mathbb{B}_i}(T),T\in\mathbb{E}/\mathbb{B})$ are indexed by the cosets $T\in\mathbb{E}/\mathbb{B}$, and we use $[\mathbi{z}+\mathbb{B}]$ to denote the coset containing $\mathbi{z}$ (see line 3).

\begin{algorithm}
\caption{The \texttt{Aggregation} function for BSC}    \label{alg:BSCAg}
\textbf{Input:} $y,\hat{y}_{/\mathbb{B}_1},\hat{y}_{/\mathbb{B}_2}\dots,\hat{y}_{/\mathbb{B}_{n-1}}$

\textbf{Output:} $y$

\vspace*{0.05in}
\begin{algorithmic}[1]
\State Initialize $(\texttt{changevote}(\mathbi{z}), \mathbi{z}\in \{0,1\}^m)$ as an all-zero vector indexed by $\mathbi{z}\in \{0,1\}^m$
\State $n \gets 2^m$
\State $\texttt{changevote}(\mathbi{z}) \gets \sum_{i=1}^{n-1} \mathbbm{1}[y_{/\mathbb{B}_i}([\mathbi{z}+\mathbb{B}_i]) \neq \hat{y}_{/\mathbb{B}_i} ([\mathbi{z}+\mathbb{B}_i])]$ for each $\mathbi{z}\in \{0,1\}^m$
\State $y(\mathbi{z}) \gets y(\mathbi{z}) \oplus \mathbbm{1}[\texttt{changevote}(\mathbi{z})>\frac{n-1}{2}]$ 
for each $\mathbi{z}\in \{0,1\}^m$
\Comment{Here addition is over $\mathbb{F}_2$}
\State \textbf{return} $y$
\end{algorithmic}
\end{algorithm}

From line 3, we can see that the maximal possible value of $\texttt{changevote}(\mathbi{z})$ for each $\mathbi{z}\in\mathbb{E}$ is $n-1$.
Therefore the condition $\texttt{changevote}(\mathbi{z})>\frac{n-1}{2}$ on line 4 can indeed be viewed as a majority vote. As discussed in Section \ref{sect:spec}, this algorithm can be viewed as one step of the power iteration method to find the eigenvector of a matrix built from the quotient code decoding. 

In Algorithms~\ref{alg:highlvl}--\ref{alg:BSCAg}, we write the pseudo codes in a mathematical fashion for the ease of understanding. In Appendix~\ref{ap:newd}, we present another version of the \texttt{RPA\_RM} function in a program language fashion. 

\begin{proposition}\label{prop:cml}
The complexity of Algorithm~\ref{alg:highlvl} is $O(n^r\log n)$ in sequential implementation and $O(n^2)$ in parallel implementation with $O(n^r)$ processors.
\end{proposition}
In Section~\ref{sect:PandA}, we further discuss options to reduce the computation time by using fewer subspaces in the projection step. 
\begin{proof}
We prove by the induction on the order of the RM code. To establish the base case, observe that the complexity of decoding first-order RM codes using Fast Hadamard Transform (FHT) \cite{Green66,Macwilliams77} is $O(n\log n)$. 
Now we assume the proposition holds for decoding $(r-1)$-th order RM codes and prove the inductive step.
Clearly, the complexity of Algorithm~\ref{alg:highlvl} is determined by the complexity of the recursive decoding step on line 3. By induction hypothesis, the complexity of decoding each $y_{/\mathbb{B}_i}$ is $O(n^{r-1}\log n)$. Since there are $n-1$ one-dimensional subspaces $\mathbb{B}_1,\mathbb{B}_2,\dots,\mathbb{B}_{n-1}$, the complexity of Algorithm~\ref{alg:highlvl} is indeed $O(n^r\log n)$.
\end{proof}

In the next proposition, we show that whether Algorithm~\ref{alg:highlvl} outputs the correct codeword or not is independent of the transmitted codeword and only depends on the error pattern imposed by the BSC channel.
\begin{proposition} \label{prop:ag1in}
Let $c\in\cR\cM(m,r)$ be a codeword of the RM code. Let $e=(e(\mathbi{z}), \mathbi{z}\in\mathbb{E})$ be the error vector imposed on $c$ by the BSC channel, and the output vector of the BSC channel is $y=c+e$. Denote the decoding result as $\hat{c}=\texttt{RPA\_RM}(y,m,r,N_{\max})$. Then the indicator function of decoding error $\mathbbm{1}[\hat{c}\neq c]$ is independent of the choice of $c$ and only depends on the error vector $e$.
\end{proposition}
Notice that we use maximal likelihood decoder for first-order RM code, and the proposition can be proved by induction on the order of the RM code\footnote{See the proof of Proposition~\ref{prop:eqde} for a rigorous argument. The ideas of the proofs of these two propositions are exactly the same.}.
This proposition is useful for simulations because we can simply transmit the all-zero codeword over the BSC channel to measure the decoding error probability.

\subsection{Spectral interpretations of Algorithm~\ref{alg:BSCAg}}   \label{sect:spec}}

Algorithm~\ref{alg:BSCAg} can be viewed as a one-step power iteration of a spectral algorithm. More precisely, observe that $\hat{y}_{/\mathbb{B}_1},\hat{y}_{/\mathbb{B}_2}\dots,\hat{y}_{/\mathbb{B}_{n-1}}$ contain the estimates of $c(\mathbi{z}) \oplus c(\mathbi{z}')$ for all $\mathbi{z} \neq \mathbi{z}'$, where $c=(c(\mathbi{z}), \mathbi{z}\in\mathbb{E})$ is the transmitted (true) codeword. 
\{We denote the estimate of $c(\mathbi{z}) \oplus c(\mathbi{z}')$
as $\hat{y}_{\mathbi{z},\mathbi{z}'}$.
Suppose for the moment that we want to find a vector $\hat{y}=(\hat{y}(\mathbi{z}), \mathbi{z}\in\mathbb{E})\in\{0,1\}^n$ to agree with as many estimates of these sums as possible, i.e., we want to find a vector $\hat{y}$ to maximize
$$
\left|\{(\mathbi{z},\mathbi{z}'):\mathbi{z}\neq\mathbi{z}',
\hat{y}(\mathbi{z})\oplus \hat{y}(\mathbi{z}')=\hat{y}_{\mathbi{z},\mathbi{z}'}\} \right|.
$$
Notice that
$$
\left|\{(\mathbi{z},\mathbi{z}'):\mathbi{z}\neq\mathbi{z}',
\hat{y}(\mathbi{z})\oplus \hat{y}(\mathbi{z}')=\hat{y}_{\mathbi{z},\mathbi{z}'}\} \right| +
\left|\{(\mathbi{z},\mathbi{z}'):\mathbi{z}\neq\mathbi{z}',
\hat{y}(\mathbi{z})\oplus \hat{y}(\mathbi{z}')\neq\hat{y}_{\mathbi{z},\mathbi{z}'}\} \right| =
n(n-1).
$$
Therefore,
$$
\sum_{\mathbi{z}\neq\mathbi{z}'}
(-1)^{\hat{y}(\mathbi{z})+ \hat{y}(\mathbi{z}')+\hat{y}_{\mathbi{z},\mathbi{z}'}}
= 2 \left|\{(\mathbi{z},\mathbi{z}'):\mathbi{z}\neq\mathbi{z}',
\hat{y}(\mathbi{z})\oplus \hat{y}(\mathbi{z}')=\hat{y}_{\mathbi{z},\mathbi{z}'}\} \right| -n(n-1) .
$$
Thus our task is equivalent to find 
\begin{equation} \label{eq:zgz}
\argmax_{\hat{y}\in\{0,1\}^n}
\sum_{\mathbi{z}\neq\mathbi{z}'}
(-1)^{\hat{y}(\mathbi{z})+ \hat{y}(\mathbi{z}')+\hat{y}_{\mathbi{z},\mathbi{z}'}} .
\end{equation}
Given a vector $\hat{y}\in\{0,1\}^n$, we define another vector $\hat{u}\in\{-1,1\}^n$ by setting 
$\hat{u}(\mathbi{z}):=(-1)^{\hat{y}(\mathbi{z})}$ for all $\mathbi{z}\in\mathbb{E}$. In order to find the maximizing vector $\hat{y}$ in \eqref{eq:zgz}, it suffices to find
\begin{equation}\label{eq:fuyu}
\argmax_{\hat{u}\in\{-1,1\}^n}
\sum_{\mathbi{z}\neq\mathbi{z}'}
(-1)^{\hat{y}_{\mathbi{z},\mathbi{z}'}} 
\hat{u}(\mathbi{z}) \hat{u}(\mathbi{z}').
\end{equation}
Now we build an $n\times n$ matrix $A$ from  $\{\hat{y}_{\mathbi{z},\mathbi{z}'}:\mathbi{z},\mathbi{z}'\in\bE,\mathbi{z}\neq\mathbi{z}'\}$ as follows: The rows and columns of $A$ are indexed by $\mathbi{z}\in\mathbb{E}$, and we set the entry 
$$
A_{\mathbi{z},\mathbi{z}'}:=
\left\{
\begin{array}{cc}
(-1)^{\hat{y}_{\mathbi{z},\mathbi{z}'}} &
\text{~if~} \mathbi{z}\neq\mathbi{z}'  \\
0 & \text{~if~} \mathbi{z}=\mathbi{z}'
\end{array}
\right. ,
$$
i.e., for $\mathbi{z}\neq\mathbi{z}'$ we set $A_{\mathbi{z},\mathbi{z}'}=1$ if $\hat{y}_{\mathbi{z},\mathbi{z}'}=0$, and $A_{\mathbi{z},\mathbi{z}'}=-1$ if $\hat{y}_{\mathbi{z},\mathbi{z}'}=1$.
Under this definition, the optimization problem \eqref{eq:fuyu} becomes
\begin{equation} \label{sbchou}
\argmax_{\hat{u}\in\{-1,1\}^n}
\sum_{\mathbi{z}\neq\mathbi{z}'}
A_{\mathbi{z},\mathbi{z}'} 
\hat{u}(\mathbi{z}) \hat{u}(\mathbi{z}')
= \argmax_{\hat{u}\in\{1,-1\}^n} \hat{u}^T A \hat{u}.
\end{equation}
It is well known that this combinatorial optimization problem is NP-hard.
In practice, people usually use 
 the following spectral relaxation to obtain approximate solution:
$$
\argmax_{\hat{u}\in\mathbb{R}^n, \|\hat{u}\|^2=n} \hat{u}^T A \hat{u}.
$$
It is well known that the solution to this relaxed optimization problem is the eigenvector corresponding to the largest eigenvalue of $A$.
One way to find this eigenvector is to use the power iteration method: pick some vector $v$ (e.g., at random), then $A^t v$ converges to this eigenvector when $t$ is large enough.\footnote{Assume the largest eigenvalue has largest magnitude.} 
After rescaling $A^t v$ to make $\|A^t v\|^2=n$, we obtain the maximizing vector $\tilde{u}=A^t v$ in the relaxed optimization problem. In order to obtain the solution to the original optimization problem in \eqref{sbchou}, we only need to look at the sign of each coordinate of $\tilde{u}$: If $\tilde{u}(\mathbi{z})>0$, then we set $\hat{u}(\mathbi{z})=1$, and if $\tilde{u}(\mathbi{z})<0$, then we set $\hat{u}(\mathbi{z})=-1$. In this way, we obtain the vector $\hat{u}$ that serves as our approximate solution to \eqref{sbchou}.
To summarize, our approximate solution to \eqref{sbchou} is $\hat{u}=\sign(A^t v)$, where $v$ is some random vector and $t$ is some large enough integer.

{Let us denote the output vector of Algorithm~\ref{alg:BSCAg} as $\overline{y}$, and we define another vector $\overline{u}$ as $\overline{u}(\mathbi{z})=(-1)^{\overline{y}(\mathbi{z})}$ for all $\mathbi{z}\in\mathbb{E}$.
For the original received vector $y$, we also define a vector $u$ as $u(\mathbi{z})=(-1)^{y(\mathbi{z})}$ for all $\mathbi{z}\in\mathbb{E}$.
The main observation in this subsection is that 
\begin{equation}\label{eq:brless}
\overline{u}=\sign(A u),
\end{equation}
i.e., the output of Algorithm~\ref{alg:BSCAg} is in fact the same as a one-step power iteration of the spectral algorithm with the original received vector $u$ playing the role of vector $v$ above.
It is also easy to see why \eqref{eq:brless} holds: According to \eqref{eq:brless}, $\overline{u}(\mathbi{z})=1$ if
$\sum_{\mathbi{z}'\neq \mathbi{z}}
(-1)^{\hat{y}_{\mathbi{z},\mathbi{z}'}\oplus y(\mathbi{z}')}>0$ and $\overline{u}(\mathbi{z})=-1$ otherwise. This is equivalent to saying that $\overline{y}(\mathbi{z})=0$ if
$|\{\mathbi{z}':\mathbi{z}'\neq \mathbi{z}, \hat{y}_{\mathbi{z},\mathbi{z}'}\oplus y(\mathbi{z}')=0\}|>\frac{n-1}{2}$ and $\overline{u}(\mathbi{z})=1$ otherwise. Clearly, the vector $\overline{y}$ given by this rule is exactly the same as the output vector of Algorithm~\ref{alg:BSCAg}.
}

We tried to use the power-iteration method in the \texttt{Aggregation} function for more than one step. However, the performance does not improve over the current version of  \texttt{Aggregation} function based on majority vote. This is because in the spectral method above we tried our best to agree with $\hat{y}_{/\mathbb{B}_1},\hat{y}_{/\mathbb{B}_2}$,$\dots$, $\hat{y}_{/\mathbb{B}_{n-1}}$, ignoring the original channel output $y$, and many of these are very noisy measurements.



\vspace*{0.2in}
\section{Decoding algorithm for general binary-input memoryless channels}  \label{sect:gen}
The decoding algorithm in the previous section only works for the BSC. In this section, we will present a natural extension of Algorithm~\ref{alg:highlvl} that works for any binary-input memoryless channels, and this new algorithm is based on LLRs (see \eqref{eq:defllr}).
Similarly to Algorithm~\ref{alg:highlvl}, this new algorithm is also defined recursively, i.e., we first assume that we know how to decode $(r-1)$-th order Reed-Muller code, and then we use it to decode the $r$-th order Reed-Muller code.
To begin with, note that the soft-decision FHT decoder \cite{Be86} allows us to decode the first order RM code efficiently for general binary-input channels. The soft-decision FHT decoder is based on LLR, and the complexity is also $O(n\log n)$, the same as the hard-decision FHT decoder.

For completeness, we recap the FHT decoder in \cite{Be86} for first order RM codes.
We still use $c=(c(\mathbi{z}), \mathbi{z}\in\mathbb{E})$ to denote the transmitted (true) codeword and $y=(y(\mathbi{z}), \mathbi{z}\in\mathbb{E})$ to denote the corresponding channel output. 
Given the output vector $y$, the ML decoder for first order RM codes aims to find $c\in\cR\cM(m,1)$ to maximize
$
\prod_{\mathbi{z}\in\mathbb{E}} W(y(\mathbi{z})| c(\mathbi{z})).
$
This is equivalent to maximizing the following quantity:
$$
\prod_{\mathbi{z}\in\mathbb{E}} \frac{W(y(\mathbi{z})| c(\mathbi{z}))}
{\sqrt{W(y(\mathbi{z})|0) W(y(\mathbi{z})|1)}},
$$
which is further equivalent to maximizing
\begin{equation}\label{eq:llr}
\sum_{\mathbi{z}\in\mathbb{E}} \ln \Big(
\frac{W(y(\mathbi{z})| c(\mathbi{z}))}
{\sqrt{W(y(\mathbi{z})|0) W(y(\mathbi{z})|1)}} \Big).
\end{equation}
Notice that the codeword $c$ is a binary vector. Therefore,
$$
\ln \Big( \frac{W(y(\mathbi{z})| c(\mathbi{z}))}
{\sqrt{W(y(\mathbi{z})|0) W(y(\mathbi{z})|1)}} \Big)
=\left\{ \begin{array}{cc}
\frac{1}{2}\LLR(y(\mathbi{z})) & \mbox{if~} c(\mathbi{z})=0 \vspace*{0.05in} \\ 
-\frac{1}{2}\LLR(y(\mathbi{z})) & \mbox{if~} c(\mathbi{z})=1
\end{array} \right.  .
$$
From now on we will use the shorthand notation
$$
L(\mathbi{z}):= \LLR(y(\mathbi{z})),
$$
and the formula in \eqref{eq:llr} can be written as
\begin{equation} \label{eq:mxr}
\frac{1}{2} \sum_{\mathbi{z}\in\mathbb{E}} \Big( (-1)^{c(\mathbi{z})} L(\mathbi{z}) \Big),
\end{equation}
so we want to find $c\in\cR\cM(m,1)$ to maximize this quantity.

By definition, every $c\in\cR\cM(m,1)$ corresponds to a polynomial in $\mathbb{F}_2[Z_1,Z_2,\dots,Z_m]$ of degree one, so we can write every codeword $c$ as a polynomial $u_0+\sum_{i=1}^m u_i Z_i$.
In this way, we have $c(\mathbi{z})=u_0+\sum_{i=1}^m u_i z_i$, where $z_1,z_2,\dots,z_m$ are the coordinates of the vector $\mathbi{z}$.
Now our task is to find $u_0,u_1,u_2,\dots,u_m\in \mathbb{F}_2$ to maximize
\begin{equation} \label{eq:fht}
\sum_{\mathbi{z}\in\mathbb{E}} \Big( (-1)^{u_0+\sum_{i=1}^m u_i z_i} L(\mathbi{z}) \Big)
=(-1)^{u_0} \sum_{\mathbi{z}\in\mathbb{E}} \Big( (-1)^{\sum_{i=1}^m u_i z_i} L(\mathbi{z}) \Big).
\end{equation}
For a binary vector $\mathbi{u}=(u_1,u_2,\dots,u_m)\in\mathbb{E}$, we define
$$
\hat{L}(\mathbi{u}):= \sum_{\mathbi{z}\in\mathbb{E}} \Big( (-1)^{\sum_{i=1}^m u_i z_i} L(\mathbi{z}) \Big).
$$
Clearly, to find the maximizer of \eqref{eq:fht}, we only need to calculate $\hat{L}(\mathbi{u})$ for all $\mathbi{u}\in\mathbb{E}$, but the vector $(\hat{L}(\mathbi{u}),\mathbi{u}\in\mathbb{E})$ is exactly the Hadamard Transform of the vector $(L(\mathbi{z}),\mathbi{z}\in\mathbb{E})$, so it can be calculated using the Fast Hadamard Transform with complexity $O(n\log n)$.
Once we know the values of $(\hat{L}(\mathbi{u}),\mathbi{u}\in\mathbb{E})$, we can find $\mathbi{u}^*=(u_1^*,u_2^*,\dots,u_m^*)\in\mathbb{E}$ that maximizes $|\hat{L}(\mathbi{u})|$.
If $\hat{L}(\mathbi{u}^*)>0$, then the decoder outputs the codeword corresponding to $u_0^*=0,u_1^*,u_2^*,\dots,u_m^*$. Otherwise, the decoder outputs the codeword corresponding to $u_0^*=1,u_1^*,u_2^*,\dots,u_m^*$.
This completes the description of how to decode the first order RM codes for general channels.

The next problem is how to extend \eqref{eq:qb} in the general setting.
The purpose of \eqref{eq:qb} is mapping two output symbols $(y(\mathbi{z}), \mathbi{z}\in T)$ whose indices are in the same coset $T\in\mathbb{E}/\mathbb{B}$ to one symbol. In this way, we reduce the $r$-th order RM code to an $(r-1)$-th order RM code.
For BSC, this mapping is simply the addition in $\mathbb{F}_2$.
The sum $y_{/\mathbb{B}}(T)$ can be interpreted as an estimate of 
$c_{/\mathbb{B}}(T)$, where $c$ is the transmitted (true) codeword. 
In other words, 
$$
\bP \big( Y_{/\mathbb{B}}(T) = c_{/\mathbb{B}}(T) \big) 
> \bP \big( Y_{/\mathbb{B}}(T) = c_{/\mathbb{B}}(T) \oplus 1 \big),
$$
where $Y$ is the channel output random vector.

For general channels, we also want to estimate $c_{/\mathbb{B}}(T)$ based on the LLRs $(L(\mathbi{z}), \mathbi{z}\in T)$. More precisely, given $(y(\mathbi{z}), \mathbi{z}\in T)$, or equivalently given $(L(\mathbi{z}), \mathbi{z}\in T)$, we would like to calculate the following LLR:
$$
L_{/\mathbb{B}}(T):=
\ln \Big( \frac{\bP\big( Y(\mathbi{z})= y(\mathbi{z}), \mathbi{z}\in T \big| c_{/\mathbb{B}}(T)=0\big)}{\bP\big( Y(\mathbi{z})= y(\mathbi{z}), \mathbi{z}\in T \big| c_{/\mathbb{B}}(T)=1\big)} \Big).
$$

We will make use of the following simple property of RM codes to calculate this LLR.
\begin{lemma}  \label{lm:soez}
Suppose that $r\ge 1$. Let $C$ be a random codeword chosen uniformly from $\cR\cM(m,r)$, and let $\mathbi{z}$ and $\mathbi{z}'$ be two distinct vectors in $\bE$. Then the two coordinates 
$(C(\mathbi{z}),C(\mathbi{z}'))$ of the random codeword $C$ have i.i.d. Bernoulli-$1/2$ distribution.
\end{lemma}
\begin{proof}
Define the following four sets
\begin{align*}
\cA(0,0):=\{c\in\cR\cM(m,r):c(\mathbi{z})=c(\mathbi{z}')=0\}, \quad
\cA(0,1):=\{c\in\cR\cM(m,r):c(\mathbi{z})=0,c(\mathbi{z}')=1\},  \\
\cA(1,0):=\{c\in\cR\cM(m,r):c(\mathbi{z})=1, c(\mathbi{z}')=0\}, \quad
\cA(1,1):=\{c\in\cR\cM(m,r):c(\mathbi{z})=c(\mathbi{z}')=1\}.
\end{align*}
To prove this lemma, we only need to show that $|\cA(0,0)|=|\cA(0,1)|=|\cA(1,0)|=|\cA(1,1)|$.
Since RM code is linear and the all one vector is a codeword of RM codes, the marginal distribution of the coordinate $C(\mathbi{z})$ is Bernoulli-$1/2$ for every $\mathbi{z}\in\bE$. Thus we have
\begin{equation} \label{eq:ouu}
|\cA(0,0)|+|\cA(0,1)|
=|\cA(1,0)|+|\cA(1,1)|, \quad\quad
|\cA(0,0)|+|\cA(1,0)|
=|\cA(0,1)|+|\cA(1,1)|.
\end{equation}
Now take $\mathbi{z}=(z_1,\dots,z_m)$ and $\mathbi{z}'=(z_1',\dots,z_m')$ such that $\mathbi{z}\neq \mathbi{z}'$. Then there exists $i\in[m]$ such that $z_i\neq z_i'$.
Since we assume that $r\ge 1$, $\cR\cM(m,r)$ contains the evaluation vector of the degree-$1$ monomial $Z_i$. We denote this evaluation vector as $v$, and we know that $v(\mathbi{z})\neq v(\mathbi{z}')$. Without loss of generality, assume that $v(\mathbi{z})=0$ and $v(\mathbi{z}')=1$.
Then we have\footnote{For a set $\cA$ and a vector $v$, we define the set $\cA+v:=\{a+v:a\in\cA\}$.} $\cA(0,0)+v \subseteq \cA(0,1)$, so $|\cA(0,0)|\le |\cA(0,1)|$. Conversely, we also have $\cA(0,1)+v \subseteq \cA(0,0)$, so $|\cA(0,1)|\le |\cA(0,0)|$. Therefore, $|\cA(0,1)|= |\cA(0,0)|$. Similarly, we can also show that $|\cA(1,1)|= |\cA(1,0)|$. Taking these into \eqref{eq:ouu}, we obtain that $|\cA(0,0)|=|\cA(0,1)|=|\cA(1,0)|=|\cA(1,1)|$, which completes the proof of the lemma.
\end{proof}

Now we can calculate $L_{/\mathbb{B}}(T)$ using the following model: Suppose that $S_1$ and $S_2$ are i.i.d. Bernoulli-$1/2$ random variables, and we transmit them over two independent copies of the channel $W:\{0,1\}\to\cW$. The corresponding channel output random variables are denoted as $X_1$ and $X_2$, respectively.
Then for $x_1,x_2\in\cW$,
\begin{align*}
& \ln \Big( \frac{\bP(X_1=x_1, X_2=x_2 | S_1+S_2=0 )}{\bP(X_1=x_1,X_2=x_2 | S_1+S_2=1 )} \Big)
=  \ln \Big( \frac{\bP(X_1=x_1, X_2=x_2 , S_1+S_2=0 )}{\bP(X_1=x_1, X_2=x_2 , S_1+S_2=1 )} \Big) \\
= & \ln \Big( \frac{\bP(X_1=x_1, X_2=x_2 , S_1=0,S_2=0 )+\bP(X_1=x_1, X_2=x_2 , S_1=1,S_2=1 )}
{\bP(X_1=x_1, X_2=x_2 , S_1=0,S_2=1 ) + \bP(X_1=x_1, X_2=x_2 , S_1=1,S_2=0 )} \Big)   \\
= & \ln \Big( \frac{ \frac{1}{4}W(x_1|0)W(x_2|0)+\frac{1}{4}W(x_1|1)W(x_2|1) }
{\frac{1}{4}W(x_1|0)W(x_2|1) + \frac{1}{4}W(x_1|1)W(x_2|0)} \Big)  
=  \ln \Big( \frac{ \frac{W(x_1|0)W(x_2|0)}{W(x_1|1)W(x_2|1)}+1 }
{\frac{W(x_1|0)}{W(x_1|1)} + \frac{W(x_2|0)}{W(x_2|1)}} \Big)  \\
= & \ln \Big( \exp \big( \LLR(x_1)+\LLR(x_2) \big) +1 \Big) - \ln \Big( \exp(\LLR(x_1))+\exp(\LLR(x_2)) \Big).
\end{align*}

Lemma~\ref{lm:soez} above allows us to replace $x_1,x_2$ with $(y(\mathbi{z}), \mathbi{z}\in T)$, and we obtain that
\begin{equation}  \label{eq:lt}
L_{/\mathbb{B}}(T)=\ln \Big( \exp \big( \sum_{\mathbi{z}\in T} L(\mathbi{z}) \big) +1 \Big) - 
\ln \Big( \sum_{\mathbi{z}\in T} \exp(L(\mathbi{z})) \Big).
\end{equation}

Now we are ready to present the decoding algorithm for general binary-input channels. 
In Algorithms~\ref{alg:genhighlvl}--\ref{alg:genAg} below, we still denote the decoding result of the $(r-1)$-th order RM code as $\hat{y}_{/\mathbb{B}}$ (see line 7 of Algorithm~\ref{alg:genhighlvl}), where 
$\hat{y}_{/\mathbb{B}}=(\hat{y}_{/\mathbb{B}}(T),T\in\mathbb{E}/\mathbb{B})$ are indexed by the cosets $T\in\mathbb{E}/\mathbb{B}$, and we use $[\mathbi{z}+\mathbb{B}]$ to denote the coset containing $\mathbi{z}$ (see line 3 of Algorithm~\ref{alg:genAg}).

Algorithm~\ref{alg:genhighlvl} is very similar to Algorithm~\ref{alg:highlvl}:
From line 8 to line 10, we compare $\hat{L}(\mathbi{z})$ with the original $L(\mathbi{z})$. If the relative difference between these two is below the threshold $\theta$ for every $\mathbi{z}\in\mathbb{E}$, then the values of $L(\mathbi{z}),\mathbi{z}\in\mathbb{E}$ change very little in this iteration, and the algorithm reaches a ``stable" state, so we can exit the for loop on line 2.
In practice, we find that $\theta=0.05$ works fairly well\footnote{The decoding error probability of this algorithm is non-increasing when we decrease the value of $\theta$, and the running time of the algorithm increases when we decrease $\theta$. Through simulations we find that the decoding error probability remains the same if we continue decreasing $\theta$ beyond $0.05$. Therefore, $\theta=0.05$ is a good choice in practice because smaller $\theta$ will only increase the running time and not decrease decoding error at all.}, and we still set the maximal number of iterations $N_{\max}=m/2$, which is the same as in Algorithm~\ref{alg:highlvl}.
On line 13, the algorithm simply produces the decoding result according to the LLR at each coordinate.

A few explanations of Algorithm~\ref{alg:genAg}: On line 3, we set
$\texttt{cumuLLR}(\mathbi{z})=\sum_{\mathbi{z}'\neq \mathbi{z}} \alpha(\mathbi{z},\mathbi{z}')L(\mathbi{z}')$, where the coefficients $\alpha(\mathbi{z},\mathbi{z}')$ can only be $1$ or $-1$.
More precisely, $\alpha(\mathbi{z},\mathbi{z}')$ is $1$ if the decoding result of the corresponding $(r-1)$th order RM code at the coset $\{\mathbi{z},\mathbi{z}'\}$ is $0$, and $\alpha(\mathbi{z},\mathbi{z}')$ is $-1$ if the decoding result at the coset $\{\mathbi{z},\mathbi{z}'\}$ is $1$.
The reason behind this assignment is simple: The decoding result at the coset $\{\mathbi{z},\mathbi{z}'\}$ is an estimate of $c(\mathbi{z})\oplus c(\mathbi{z}')$. If $c(\mathbi{z})\oplus c(\mathbi{z}')$ is more likely to be $0$, then the sign of $L(\mathbi{z})$ and $L(\mathbi{z}')$ should be the same. Here $\texttt{cumuLLR}(\mathbi{z})$ serves as an estimate of $L(\mathbi{z})$ based on all the other $L(\mathbi{z}'), \mathbi{z}'\neq \mathbi{z}$, so we assign the coefficient $\alpha(\mathbi{z},\mathbi{z}')$ to be $1$. Otherwise, if $c(\mathbi{z})\oplus c(\mathbi{z}')$ is more likely to be $1$, then the sign of $L(\mathbi{z})$ and $L(\mathbi{z}')$ should be different, so we assign the coefficient $\alpha(\mathbi{z},\mathbi{z}')$ to be $-1$.

\begin{algorithm}
\caption{The \texttt{RPA\_RM} decoding function for general binary-input memoryless channels}    \label{alg:genhighlvl}
\textbf{Input:} The LLR vector $(L(\mathbi{z}), \mathbi{z}\in \{0,1\}^m)$; the parameters of the Reed-Muller code $m$ and $r$; the maximal number of iterations $N_{\max}$; the exiting threshold $\theta$

\textbf{Output:} The decoded codeword $\hat{c}$

\vspace*{0.05in}
\begin{algorithmic}[1]
\State  $\mathbb{E} := \{0,1\}^m$
\For {$j=1,2,\dots,N_{\max}$} 
\State $L_{/\mathbb{B}_i} \gets (L_{/\mathbb{B}_i}(T),T\in\mathbb{E}/\mathbb{B}_i)$  for $i=1,2,\dots,2^m-1$                         \Comment{Projection}
\State \Comment{$L_{/\mathbb{B}_i}(T)$ is calculated from $(L(\mathbi{z}), \mathbi{z}\in\mathbb{E})$ according to \eqref{eq:lt}}     
\State $\hat{y}_{/\mathbb{B}_i} \gets \texttt{RPA\_RM}(L_{/\mathbb{B}_i},m-1,r-1,N_{\max}, \theta)$
for $i=1,2,\dots,2^m-1$ 
\Comment{Recursive decoding}
\State \Comment{If $r=2$, then we use the Fast Hadamard Transform to decode the first-order RM code}
\State $\hat{L} \gets \texttt{Aggregation}(L,\hat{y}_{/\mathbb{B}_1},\hat{y}_{/\mathbb{B}_2}\dots,\hat{y}_{/\mathbb{B}_{n-1}})$           \Comment{Aggregation}
\If {$|\hat{L}(\mathbi{z})-L(\mathbi{z})|\le \theta |L(\mathbi{z})|$ for all $\mathbi{z}\in\mathbb{E}$}
\Comment{The algorithm reaches a stable point}
\State \textbf{break}  
\EndIf
\State $L \gets \hat{L}$
\EndFor
\State $\hat{c}(\mathbi{z}) \gets \mathbbm{1}[L(\mathbi{z}) < 0]$ for each $\mathbi{z}\in\mathbb{E}$

\State \textbf{return} $\hat{c}$
\end{algorithmic}
\end{algorithm}

\begin{algorithm}
\caption{The \texttt{Aggregation} function for general binary-input memoryless channels}    \label{alg:genAg}
\textbf{Input:} $L,\hat{y}_{/\mathbb{B}_1},\hat{y}_{/\mathbb{B}_2}\dots,\hat{y}_{/\mathbb{B}_{n-1}}$

\textbf{Output:} $\hat{L}$

\vspace*{0.05in}
\begin{algorithmic}[1]
\State Initialize $(\texttt{cumuLLR}(\mathbi{z}), \mathbi{z}\in\{0,1\}^m)$ as an all-zero vector indexed by $\mathbi{z}\in\{0,1\}^m$
\State $n \gets 2^m$

\State $\texttt{cumuLLR}(\mathbi{z}) \gets \sum_{i=1}^{n-1}
\big((1-2\hat{y}_{/\mathbb{B}_i}([\mathbi{z}+\mathbb{B}_i])) L(\mathbi{z}\oplus \mathbi{z}_i) \big)$   for each $\mathbi{z}\in \{0,1\}^m$
\State \Comment{$\mathbi{z}_i$ is the nonzero element in $\mathbb{B}_i$}
\State  \Comment{$\hat{y}_{/\mathbb{B}_i}$ is the decoded codeword, so $\hat{y}_{/\mathbb{B}_i}([\mathbi{z}+\mathbb{B}_i])$ is either $0$ or $1$}

\State $\hat{L}(\mathbi{z}) \gets \frac{\texttt{cumuLLR}(\mathbi{z})} {n-1}$ for each $\mathbi{z}\in \{0,1\}^m$

\State \textbf{return} $\hat{L}$
\end{algorithmic}
\end{algorithm}

In Algorithms~\ref{alg:genhighlvl}--\ref{alg:genAg}, we write the pseudo codes in a mathematical fashion for the ease of understanding. In Appendix~\ref{ap:general}, we present another version of the \texttt{RPA\_RM} function in a program language fashion.

Following the same proof of Proposition~\ref{prop:cml}, we have the following result:
\begin{proposition}
The complexity of Algorithm~\ref{alg:genhighlvl} is $O(n^r\log n)$ in sequential implementation and $O(n^2)$ in parallel implementation with $O(n^r)$ processors.
\end{proposition}
In Section~\ref{sect:highhigh}, we present an accelerated version of the RPA algorithm for high-rate RM codes, and
in Section~\ref{sect:PandA}, we further discuss other possible options to reduce the computation time by using fewer subspaces in the projection step.

Similarly to Proposition~\ref{prop:ag1in}, we can also show that the decoding error probability of Algorithm~\ref{alg:genhighlvl} is independent of the transmitted codeword for binary-input memoryless symmetric (BMS) channels.

\begin{definition}[BMS channel]  \label{def:bms}
We say that a memoryless channel $W:\{0,1\}\to \cW$ is a BMS channel if there is a permutation $\pi$ of the output alphabet $\cW$ such that $\pi^{-1}=\pi$ and $W(x|1)=W(\pi(x)|0)$ for all $x\in\cW$.
\end{definition}

\begin{proposition} \label{prop:eqde}
Let $W:\{0,1\}\to \cW$ be a BMS channel. Let $c_1$ and $c_2$ be two codewords of $\cR\cM(m,r)$. Let $Y_1$ and $Y_2$ be the (random) channel outputs of transmitting $c_1$ and $c_2$ over $n=2^m$ independent copies of $W$, respectively.
Let $L^{(1)}$ and $L^{(2)}$ be the LLR vectors corresponding to $Y_1$ and $Y_2$, respectively\footnote{$Y_1$ and $Y_2$ are random vectors, and the randomness comes from the channel noise. As a result, $L^{(1)}$ and $L^{(2)}$ are also random vectors.}.
Then for any $c_1,c_2\in\cR\cM(m,r)$, we have
$$
\mathbb{P}(\texttt{RPA\_RM}(L^{(1)},m,r,N_{\max},\theta)\neq c_1)
= \mathbb{P}(\texttt{RPA\_RM}(L^{(2)},m,r,N_{\max},\theta)\neq c_2).
$$
\end{proposition}

The proof is given in Appendix~\ref{ap:eqde}.
Similarly to Proposition~\ref{prop:ag1in},
this proposition is also very useful for simulations because we can simply transmit the all-zero codeword over the BMS channel $W$ to measure the decoding error probability.

In the last part of this section, we present the list decoding version of the \texttt{RPA\_RM} function. The main idea is already explained in Section~\ref{sect:highlist}. Here we only write down the pseudo code of the list decoding version. Note that the purpose of line 8 is to make sure that $\hat{c}^{(\mathbi{u})}$ is a codeword of RM code, which is not always true for the decoding result of the \texttt{RPA\_RM} function.

Finally, we present the following proposition on the memory requirement for sequential implementation of RPA decoder. A remarkable thing here is that the memory requirement for the list decoding version of RPA algorithm is $5n$, which is independent of the list size, in contrast to SCL decoder of polar codes.
\begin{proposition}
The memory needed for sequential implementation of the RPA decoder without list decoding is no more than $4n$, and the memory needed for sequential implementation of the RPA decoder with list decoding is no more than $5n$, where $n$ is the code length. Note that the memory requirement for list decoding version does not depend on the list size.
\end{proposition}
\begin{proof}
As we mentioned above, Algorithm~\ref{alg:genhighlvl} is written in  compact fashion for the ease of understanding, but it is not space-efficient in practical implementation.  
The version that we really implemented in practice and used for simulations is Algorithm~\ref{alg:general} in Appendix~\ref{ap:general}, and our analysis of space complexity is based on Algorithm~\ref{alg:general}.

The most important difference between Algorithm~\ref{alg:general} and Algorithm~\ref{alg:genhighlvl} is that in Algorithm~\ref{alg:genhighlvl} we first finish all the recursive decoding and then perform the aggregation step; while in Algorithm~\ref{alg:general} the recursive decoding step and the aggregation step are interleaved together, and in this way we can save huge amount of memory compared to Algorithm~\ref{alg:genhighlvl}.

We start with RPA decoder without list decoding, and we prove by induction on $r$, the order of the RM code. 
For the base case of $r=1$, the claim clearly holds. Now assume that the claim holds for all RM codes with order $<r$ and we prove it for order $r$.
In Algorithm~\ref{alg:general}, we need $n$ floating number positions to store the LLR vector and another $n$ floating number positions to store the \texttt{cumuLLR} vector. Then we project onto the cosets of each one-dimensional subspace sequentially. For each projected codeword, we need to decode a RM code with length $n/2$ and order $r-1$. By induction hypothesis, this take $4*n/2=2n$ floating number positions. Therefore in total we need $n+n+2n=4n$ floating number positions. This establishes the inductive step and completes the proof for the non-list-decoding version.

The memory requirement for list decoding version follows directly from that of the vanilla version: Since we perform list decoding sequentially, i.e., we only decode one list at a time, the only extra memory we need in the list decoding version is the $n$ floating number positions that is used to store currently best known decoding result. Therefore, the space complexity for the list decoding version is $5n$.
\end{proof}

\begin{algorithm}  
\caption{The \texttt{RPA\_LIST} decoding function for general binary-input memoryless channels}   
\label{alg:list}
\textbf{Input:} The LLR vector $(L(\mathbi{z}), \mathbi{z}\in \{0,1\}^m)$; the parameters of the Reed-Muller code $m$ and $r$; the maximal number of iterations $N_{\max}$; the exiting threshold $\theta$; the list size $2^t$

\textbf{Output:} The decoded codeword $\hat{c}$

\vspace*{0.05in}
\begin{algorithmic}[1]
\State $\tilde{L}\gets L$
\State  $(\mathbi{z}_1,\mathbi{z}_2,\dots,\mathbi{z}_t) \gets$ indices of the $t$ smallest entries in $(|L(\mathbi{z})|, \mathbi{z}\in \{0,1\}^m)$
\State \Comment{$\mathbi{z}_i\in\{0,1\}^m$ for all $i=1,2,\dots,t$}
\State $L_{\max} \gets 2\max(|L(\mathbi{z})|, \mathbi{z}\in \{0,1\}^m)$
\For {each $\mathbi{u}\in\{L_{\max},-L_{\max}\}^t$}
\State $(L(\mathbi{z}_1), L(\mathbi{z}_2), \dots, L(\mathbi{z}_t)) \gets \mathbi{u}$
\State $\hat{c}^{(\mathbi{u})} \gets \texttt{RPA\_RM}(L,m,r,N_{\max}, \theta)$
\State $\hat{c}^{(\mathbi{u})} \gets \texttt{Reedsdecoder}(\hat{c}^{(\mathbi{u})})$
\Comment{\texttt{Reedsdecoder} is the classical decoding algorithm in \cite{Reed54}}
\EndFor
\State $\mathbi{u}^* \gets \argmax_{\mathbi{u}} \sum_{\mathbi{z}\in \{0,1\}^m} \Big( (-1)^{\hat{c}^{(\mathbi{u})}(\mathbi{z})} \tilde{L}(\mathbi{z}) \Big)$
\State \Comment{This follows from \eqref{eq:mxr}. Maximization is over $\mathbi{u}\in\{L_{\max},-L_{\max}\}^t$}
\State $\hat{c}\gets \hat{c}^{(\mathbi{u}^*)}$
\State \textbf{return} $\hat{c}$
\end{algorithmic}
\end{algorithm}

{
\section{Simplified RPA algorithm for high rate RM codes}
\label{sect:highhigh}
In this section, we provide some simplified versions of the RPA decoder, which significantly accelerate the decoding process while maintaining the same (nearly optimal) decoding error probability for certain RM codes with rate $>0.5$.

As mentioned in the previous section, we can accelerate the decoding algorithm by using fewer subspaces in the projection step. Moreover, instead of using one-dimensional subspaces, in this section we propose to use a selected subsets of two-dimensional subspaces in the projection step.
In particular, we only project onto the $\binom{m}{2}$ two-dimensional subspaces spanned by two standard basis vectors of $\bE$. The standard basis vector of $\bE$ are $\mathbi{e}^{(1)},\dots,\mathbi{e}^{(m)}$, where $\mathbi{e}^{(i)}$ is defined as the vector with $1$ in the $i$th position and $0$ everywhere else.
Then we write the $\binom{m}{2}$ two-dimensional subspaces as $\{\bB_{i,j}:1\le i<j\le m\}$, where 
$$
\bB_{i,j}:=\spn(\mathbi{e}^{(i)},\mathbi{e}^{(j)}).
$$

Note that projection onto cosets of two-dimensional subspaces is different from onto that of one-dimensional subspaces: In the one-dimensional case, each coset only contains two coordinates, and we only need to combine the LLR of two coordinates to obtain the LLR of the coset, as we did in \eqref{eq:lt}. In the two-dimensional case, each coset contains four coordinates, and we need to combine the LLR of four coordinates to obtain the LLR of the coset. Fortunately,
for any RM code with order $r\ge 2$,
we can use exactly the same idea  in the proof of Lemma~\ref{lm:soez} to show that any four coordinates in a coset of a two-dimensional subspace are also independent; see the explanation in Remark~\ref{rmk:1} below. Therefore, we obtain the following counterpart of \eqref{eq:lt} for a coset $T$ of two-dimensional subspace assuming that $T=\{\mathbi{z}^{(1)},\mathbi{z}^{(2)},\mathbi{z}^{(3)},\mathbi{z}^{(4)}\}$:
\begin{equation} \label{eq:fkn}
\begin{aligned}
L_{/\mathbb{B}}(T)= &\ln \Big( \exp \big( \sum_{i=1}^4 L(\mathbi{z}^{(i)}) \big)
+\sum_{1\le i<j \le 4} \exp \big(  L(\mathbi{z}^{(i)}) + L(\mathbi{z}^{(j)}) \big)  +1 \Big) \\
& - 
\ln \Big( \sum_{i=1}^4 \exp(L(\mathbi{z}^{(i)})) 
+\sum_{i=1}^4 \exp(
\sum_{j\in[4]\setminus\{i\}} L(\mathbi{z}^{(j)})) \Big).
\end{aligned}
\end{equation}

\begin{remark} \label{rmk:1}
It is well known that for a linear code, if there is a codeword taking value $1$ at a certain coordinate, then the number of codewords taking value $1$ at this coordinate is the same as the number of codewords taking value $0$ at this coordinate. This follows directly from the linearity of the code. The proof of Lemma~\ref{lm:soez} follows from the same idea: By the linearity of code, we only need to show that for two distinct coordinates, there are different codewords in RM codes that take all four possible values $(0,0), (0,1), (1,0), (1,1)$ at these two coordinates, and this follows by noting that (i) any two distinct coordinates form a coset of a one-dimensional subspace; (ii) by definition of RM codes, restricting RM codes with order $r\ge 1$ on such cosets gives us $\cR\cM(1,1)$, which contains all 4 binary vectors of length $2$. Now in the case of two-dimensional subspace, we still use the same reasoning: By linearity of the code, we only need to show that for any $4$ coordinates that form a coset of a 2-dimensional subspace, there are different codewords in RM codes with order $r\ge 2$ that take all $2^4$ possible values $\{0,1\}^4$ at these four coordinates.
This again follows by noting that restricting RM codes with order $r\ge 2$ on such cosets gives us $\cR\cM(2,2)$, which contains all $16$ binary vectors of length $4$.
\end{remark}

\begin{algorithm}
\caption{The \texttt{Simplified\_RPA} decoding function}    \label{alg:simp}
\textbf{Input:} The LLR vector $(L(\mathbi{z}), \mathbi{z}\in \{0,1\}^m)$; the parameters of the Reed-Muller code $m$ and $r$; the maximal number of iterations $N_{\max}$; the exiting threshold $\theta$

\textbf{Output:} The decoded codeword $\hat{c}$

\vspace*{0.05in}
\begin{algorithmic}[1]
\State  $\mathbb{E} := \{0,1\}^m$
\For {$j=1,2,\dots,N_{\max}$} 
\State $L_{/\mathbb{B}_{i,j}} \gets (L_{/\mathbb{B}_{i,j}}(T),T\in\mathbb{E}/\mathbb{B}_{i,j})$  for $1\le i<j\le m$                         \Comment{Projection}
\State \Comment{$L_{/\mathbb{B}_{i,j}}(T)$ is calculated according to \eqref{eq:fkn}}     
\State $\hat{y}_{/\mathbb{B}_{i,j}} \gets \texttt{Simplified\_RPA}(L_{/\mathbb{B}_{i,j}},m-2,r-2,N_{\max}, \theta)$
for $1\le i<j\le m$ 
\State \Comment{Recursive decoding}
\State \Comment{If $r=3$, then we use the Fast Hadamard Transform to decode the first-order RM code}
\State \Comment{If $r=4$, then we use the normal RPA algorithm to decode the second-order RM code}
\State $\hat{L} \gets \texttt{Simp\_Aggregation}(L,\{\hat{y}_{/\mathbb{B}_{i,j}}:1\le i<j\le m\})$           \Comment{Aggregation}
\If {$|\hat{L}(\mathbi{z})-L(\mathbi{z})|\le \theta |L(\mathbi{z})|$ for all $\mathbi{z}\in\mathbb{E}$}
\Comment{The algorithm reaches a stable point}
\State \textbf{break}  
\EndIf
\State $L \gets \hat{L}$
\EndFor
\State $\hat{c}(\mathbi{z}) \gets \mathbbm{1}[L(\mathbi{z}) < 0]$ for each $\mathbi{z}\in\mathbb{E}$

\State \textbf{return} $\hat{c}$
\end{algorithmic}
\end{algorithm}

\begin{algorithm}
\caption{The \texttt{Simp\_Aggregation} function in the \texttt{Simplified\_RPA} algorithm}    \label{alg:simpAg}
\textbf{Input:} $L,\{\hat{y}_{/\mathbb{B}_{i,j}}:1\le i<j\le m\}$

\textbf{Output:} $\hat{L}$

\vspace*{0.05in}
\begin{algorithmic}[1]

\State Calculate $\Est_{i,j}(\mathbi{z})$ from $L$ and $\{\hat{y}_{/\mathbb{B}_{i,j}}:1\le i<j\le m\}$ according to \eqref{eq:ynt}

\State $\hat{L}(\mathbi{z})\gets
\frac{1}{\binom{m}{2}} \sum_{1\le i<j\le m} \Est_{i,j}(\mathbi{z})$   for each $\mathbi{z}\in \{0,1\}^m$

\State \textbf{return} $\hat{L}$
\end{algorithmic}
\end{algorithm}

After projecting $\cR\cM(m,r)$ onto the cosets of these two-dimensional subspaces, we will obtain RM codes with parameters $m-2$ and $r-2$, as proved in Lemma~\ref{lm:tbs}.
After decoding these $\binom{m}{2}$ projected codes $\cR\cM(m-2,r-2)$, we obtain $\{\hat{y}_{/\mathbb{B}_{i,j}}:1\le i<j\le m\}$, where 
$\hat{y}_{/\mathbb{B}_{i,j}}=(\hat{y}_{/\mathbb{B}_{i,j}}(T),T\in\mathbb{E}/\mathbb{B}_{i,j})$. Now we are ready to go to the aggregation step using both the recursive decoding result $\{\hat{y}_{/\mathbb{B}_{i,j}}:1\le i<j\le m\}$ and the original LLR vector $L$.
In particular, when decoding $c(\mathbi{z})$, the relevant coordinate in $\hat{y}_{/\mathbb{B}_{i,j}}$ is $\hat{y}_{/\mathbb{B}_{i,j}}([\mathbi{z}+\mathbb{B}_{i,j}])$, where $[\mathbi{z}+\mathbb{B}_{i,j}]$ is the coset of $\mathbb{B}_{i,j}$ that contains $\mathbi{z}$.
Now suppose that the other three vectors in $[\mathbi{z}+\mathbb{B}_{i,j}]$ apart from $\mathbi{z}$ itself are $\mathbi{z}^{(1)},\mathbi{z}^{(2)},\mathbi{z}^{(3)}$. Then from $\hat{y}_{/\mathbb{B}_{i,j}}([\mathbi{z}+\mathbb{B}_{i,j}])$ and $L(\mathbi{z}^{(1)}),L(\mathbi{z}^{(2)}),L(\mathbi{z}^{(3)})$, we obtain the following estimate of the LLR of $c(\mathbi{z})$:
\begin{equation} \label{eq:ynt}
\begin{aligned}
\Est_{i,j}(\mathbi{z})=
\ln \Big( \exp \big( \sum_{i=1}^3 L(\mathbi{z}^{(i)}) \big)
+ \sum_{i=1}^3 \exp(L(\mathbi{z}^{(i)})) \Big)  - \ln \Big( \sum_{i=1}^3 \exp(
\sum_{j\in[3]\setminus\{i\}} L(\mathbi{z}^{(j)})) 
+ 1\Big) \\
\text{if~} \hat{y}_{/\mathbb{B}_{i,j}}([\mathbi{z}+\mathbb{B}_{i,j}])=0 , \\
\Est_{i,j}(\mathbi{z})=
- \ln \Big( \exp \big( \sum_{i=1}^3 L(\mathbi{z}^{(i)}) \big)
+ \sum_{i=1}^3 \exp(L(\mathbi{z}^{(i)})) \Big)  + \ln \Big( \sum_{i=1}^3 \exp(
\sum_{j\in[3]\setminus\{i\}} L(\mathbi{z}^{(j)})) 
+ 1\Big) \\
\text{if~} \hat{y}_{/\mathbb{B}_{i,j}}([\mathbi{z}+\mathbb{B}_{i,j}])=1 .
\end{aligned}
\end{equation}
We calculate such an estimate for all pairs of $(i,j)$ such that $1\le i<j\le m$. Then finally we update the LLR of $c(\mathbi{z})$ as the average of these $\binom{m}{2}$ estimates, as follows:
$$
\hat{L}(\mathbi{z})=\frac{1}{\binom{m}{2}} \sum_{1\le i<j\le m} \Est_{i,j}(\mathbi{z}).
$$
Finally, as in all the previous sections, we iterate this decoding procedure a few times for the LLR vector to converge to a stable value.

We call the decoding algorithm proposed in this section the \texttt{Simplified\_RPA} algorithm, as opposed to the normal RPA algorithm proposed in the previous section.
Note here that in the recursive decoding procedure, i.e., when we decode $\cR\cM(m-2,r-2)$, we still use this simplified version of RPA algorithm instead of doing full projection step. Since each time we reduce $r$ by $2$, if the original $r$ is even then we will not reach the first-order RM codes. In this case, we use the normal RPA decoder when we reach the second-order RM codes.
In Algorithm~\ref{alg:simp} and Algorithm~\ref{alg:simpAg} we provide pseudo-codes for the \texttt{Simplified\_RPA} algorithm.
Note that in line 7--8 of Algorithm~\ref{alg:simp}, we distinguish between the cases of $r$ being even and $r$ being odd: For even $r$, eventually we will need to decode a second-order RM code using the normal RPA decoder while for odd $r$, we only need to decode first-order RM code in the final recursive step.
As we will show in Section~\ref{sect:sim} (see Fig.~\ref{fig:cpAWGN}), by applying the list decoding version of the \texttt{Simplified\_RPA} algorithm, we can decode $\cR\cM(7,4)$ and $\cR\cM(8,5)$ with list size no larger than $8$ such that the decoding error probability is the same as that of ML decoder. Moreover, it runs even faster than decoding lower rate codes such as $\cR\cM(8,3)$; see Table~\ref{tb:rt}.

}

\section{Simulation results} \label{sect:sim}

\subsection{Comparison with polar codes}
We run our decoding algorithm for second and third order Reed-Muller codes with code length $256,512$ and $1024$ over AWGN channels and BSCs, and we compare its performance with the recent algorithms for polar codes with the same length and dimension. 
We compare to two versions of polar codes: Polar codes with optimal CRC size and polar codes without CRC,
and we use the Successive Cancellation List (SCL) decoder introduced by Tal and Vardy \cite{Tal15} as the decoder, where we set list size to be $32$.
Note that SCL decoder with list size $32$ is one of the most widely used decoders for polar codes.

The simulation results for AWGN channels are plotted in Figure~\ref{fig:cpAWGN},
where the number of Monte Carlo trials is $10^5$.
We provide the simulation results for all RM codes with length 128 and 256, including $\cR\cM(7,2),\cR\cM(7,3),\cR\cM(7,4),\cR\cM(8,2),\cR\cM(8,3),\cR\cM(8,4),\cR\cM(8,5)$. This should give a complete picture of the performance of our decoder for all code rates. Note that we skipped $\cR\cM(7,5)$ and $\cR\cM(8,6)$ because they are extended Hamming codes, and optimal decoders are well known for these two codes.
Moreover, for certain cases the list decoding version of RPA decoding algorithm has almost the same performance as the Maximal Likelihood (ML) decoder for RM codes\footnote{We use the method in \cite{Dumer04,Dumer06a} to measure the ML lower bound: Whenever our decoder outputs a wrong codeword, we compare the posterior probability of the decoded word and that of the correct codeword. Most of the time the posterior probability of the decoded word is larger, which means that even an ML decoder will make a mistake in this case. Note that this method was also used in \cite{Tal15}.}. The performance improvement is thus in agreement with the advantages of RM codes over polar codes under ML decoding \cite{Mondelli14}. 
See Section~\ref{sect:pv} for comparisons with Dumer's recursive decoding algorithm \cite{Dumer04,Dumer06,Dumer06a}, which is the best known decoder in the literature for RM codes over AWGN channels.
Note also that the algorithm in \cite{Santi18} only applies to codes with very short code length (no larger than $128$) due to complexity constraints.

For the BSC channel, the simulation results are plotted in Figure~\ref{fig:cp}. 
The number of Monte Carlo trials is $10^5$.
We also tested in this case all the previous decoding algorithms known for RM codes, including Reed's algorithm \cite{Reed54} and the algorithm from Saptharishi-Shpilka-Volk \cite{Saptharishi17}. For these two algorithms, the decoding error probability exceeds $0.1$ for the tested parameters, so we did not include them in Figure~\ref{fig:cp} as they would not fit.
See Section~\ref{sect:pv} for comparisons with the Sidel'nikov-Pershakov algorithm \cite{Sidel92} and its variations \cite{Loidreau04,Sakkour05}.
From Figure~\ref{fig:cp}, we can clearly see that the new decoding algorithm for RM codes significantly outperforms the SCL decoder for CRC-aided polar codes.

We also compare the running time of our decoder and the SCL decoder for polar codes.
For polar codes, we use techniques from two accelerated version \cite{Balatsoukas15,Sarkis15} of the SCL decoder (in particular the ``min-sum approximation" in \cite{Balatsoukas15}) so that we can achieve a much smaller running time than the original version of SCL decoder while maintaining almost the same decoding error probability.
The results are listed in Table~\ref{tb:rt}. We can see that for second order RM codes as well as the high-rate RM codes where we use the \texttt{Simplified\_RPA} algorithm to decode, our decoder is always faster than the SCL decoder for polar codes with the same parameters. However, for third order RM codes, our decoder is slower than the SCL decoder; see Fig~\ref{fig:cpAWGN} for decoding error probability and Table~\ref{tb:rt} for running time.

\begin{figure}
\centering
\begin{subfigure}{0.3\linewidth} 
\centering
\includegraphics[width=\linewidth]{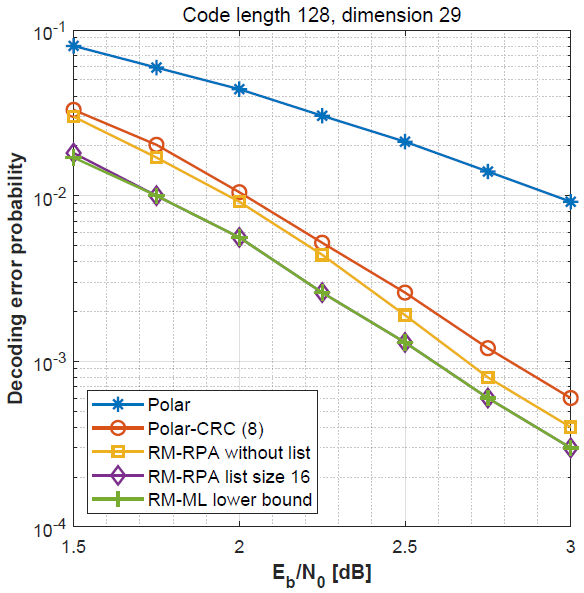}
\caption{$\cR\cM(7,2)$ v.s. polar codes}
\end{subfigure}
~\hspace*{0.1in}
\begin{subfigure}{0.3\linewidth}
\centering
\includegraphics[width=\linewidth]{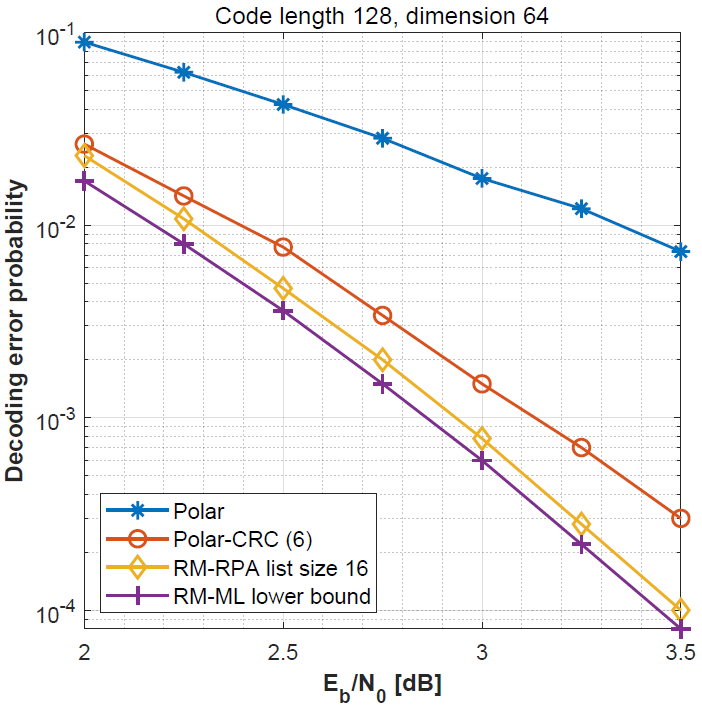}
\caption{$\cR\cM(7,3)$ v.s. polar codes}
\end{subfigure}
~\hspace*{0.1in}
\begin{subfigure}{0.3\linewidth}
\centering
\includegraphics[width=\linewidth]{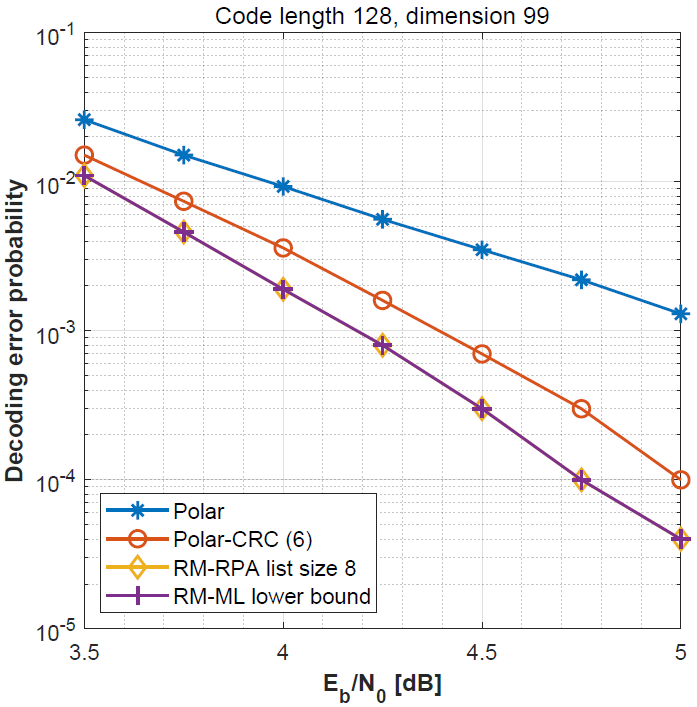}
\caption{$\cR\cM(7,4)$ v.s. polar codes}
\end{subfigure}

\vspace*{0.1in}

\begin{subfigure}{0.3\linewidth} 
\centering
\includegraphics[width=\linewidth]{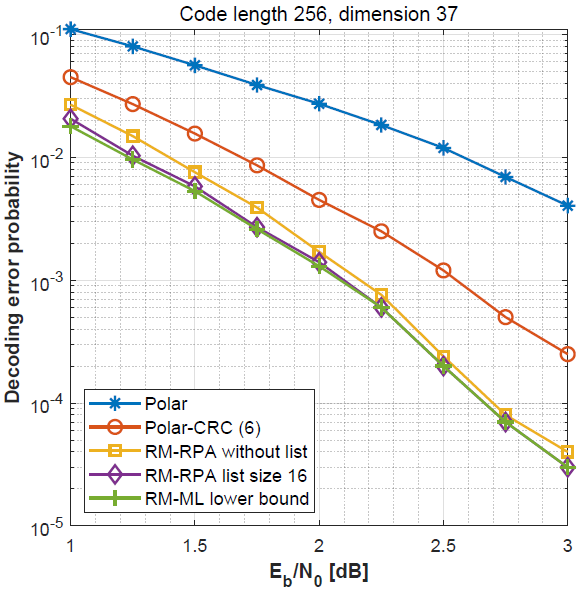}
\caption{$\cR\cM(8,2)$ v.s. polar codes}
\end{subfigure}
~\hspace*{0.1in}
\begin{subfigure}{0.3\linewidth}
\centering
\includegraphics[width=\linewidth]{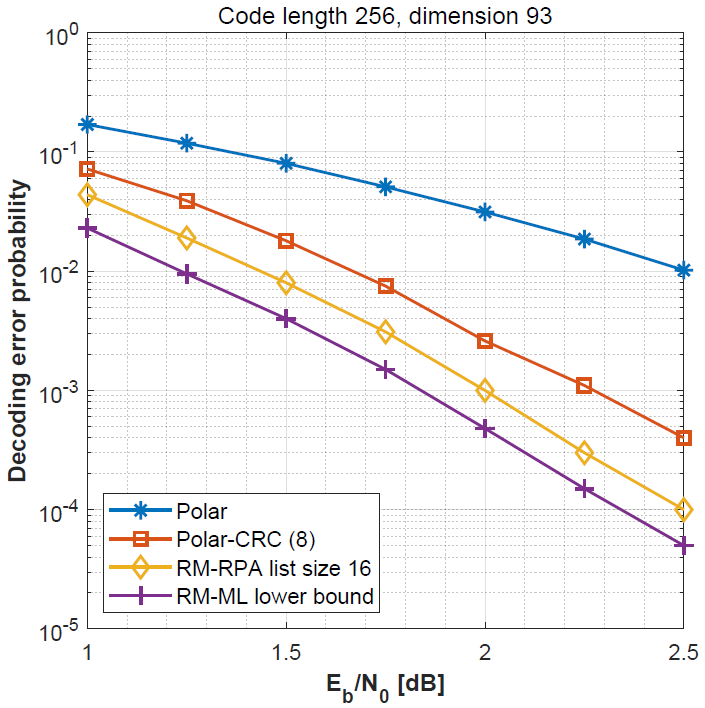}
\caption{$\cR\cM(8,3)$ v.s. polar codes}
\end{subfigure}
~\hspace*{0.1in}
\begin{subfigure}{0.3\linewidth}
\centering
\includegraphics[width=\linewidth]{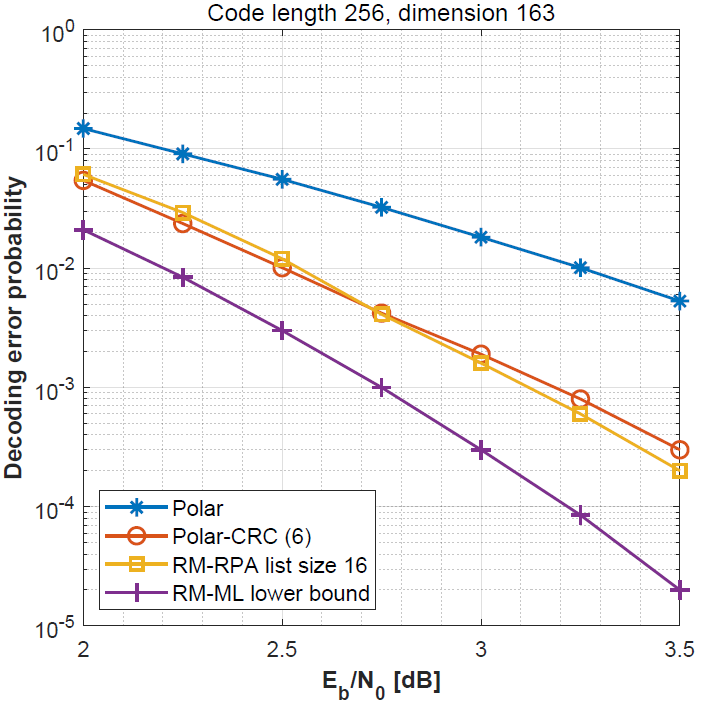}
\caption{$\cR\cM(8,4)$ v.s. polar codes}
\end{subfigure}

\vspace*{0.1in}
\begin{subfigure}{0.3\linewidth}
\centering
\includegraphics[width=\linewidth]{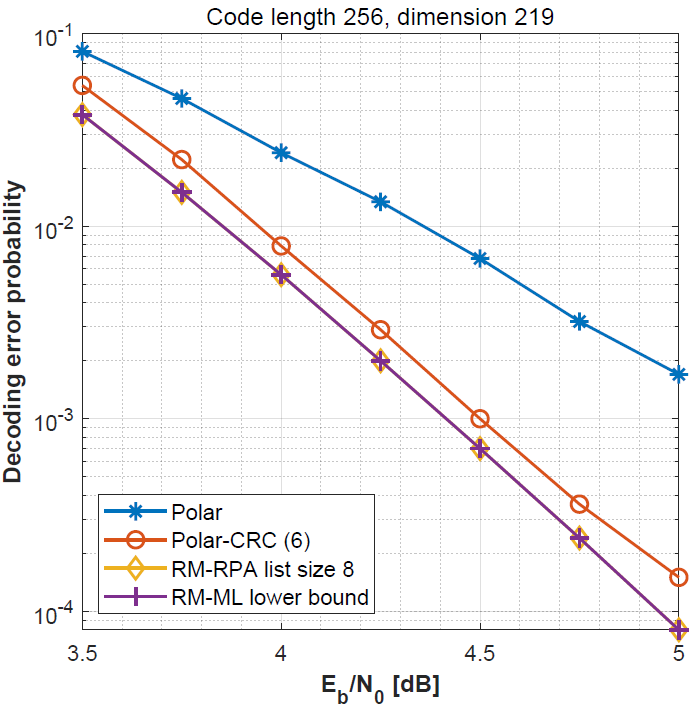}
\caption{$\cR\cM(8,5)$ v.s. polar codes}
\end{subfigure}
~\hspace*{0.1in}
\begin{subfigure}{0.3\linewidth} 
\centering
\includegraphics[width=\linewidth]{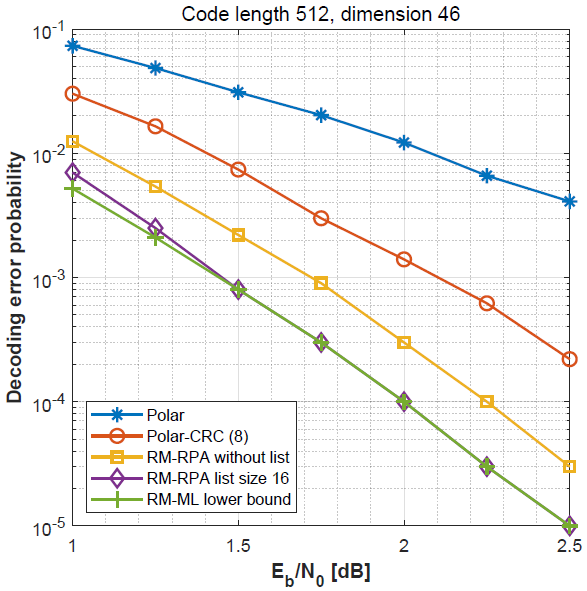}
\caption{$\cR\cM(9,2)$ v.s. polar codes}   
\end{subfigure}
~\hspace*{0.1in}
\begin{subfigure}{0.3\linewidth} 
\centering
\includegraphics[width=\linewidth]{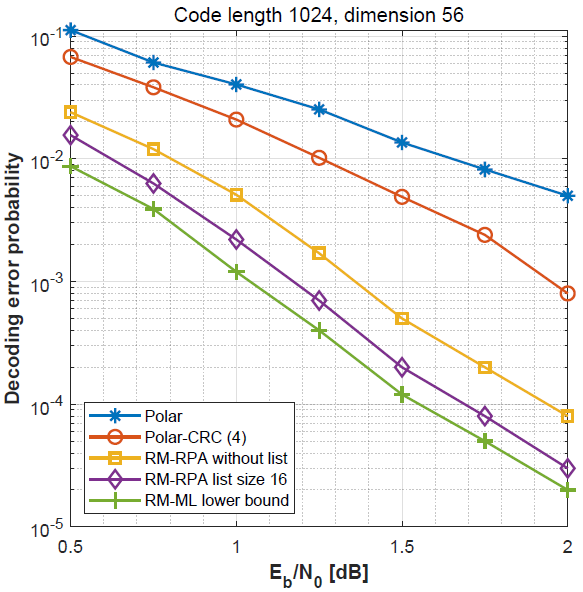}
\caption{$\cR\cM(10,2)$ v.s. polar codes}   
\end{subfigure}

\caption{Comparison between Reed-Muller codes and polar codes over AWGN channels. For $\cR\cM(7,4)$ and $\cR\cM(8,5)$, we use the \texttt{Simplified\_RPA} algorithm proposed in Section~\ref{sect:highhigh}, and for all the other RM codes, we use the normal RPA algorithm proposed in Section~\ref{sect:gen}.
For polar codes with or without CRC, we always use SCL decoder with list size $32$.
For polar codes with CRC, we test various choices of CRC length and choose the optimal one that gives the best performance. The number in the bracket after ``Polar-CRC" is the optimal CRC length that we use.}
\label{fig:cpAWGN}
\end{figure}

\begin{figure}
\centering
\begin{subfigure}{0.3\linewidth} 
\centering
\includegraphics[width=\linewidth]{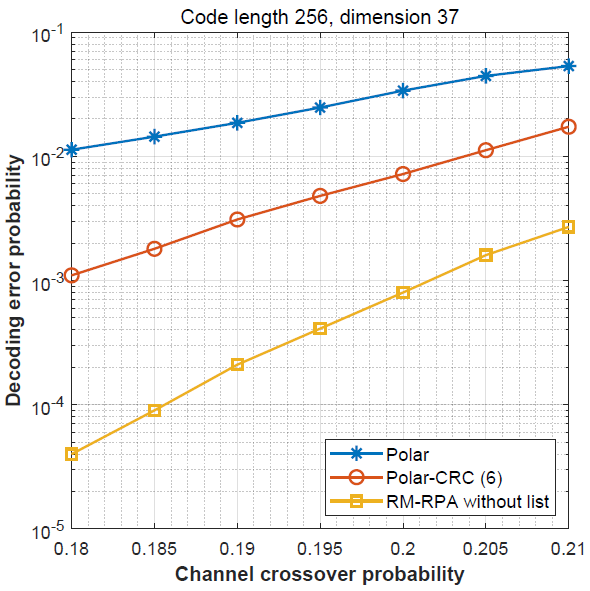}
\caption{$\cR\cM(8,2)$ vs Polar codes}
\end{subfigure}
~\hspace*{0.1in}
\begin{subfigure}{0.3\linewidth}
\centering
\includegraphics[width=\linewidth]{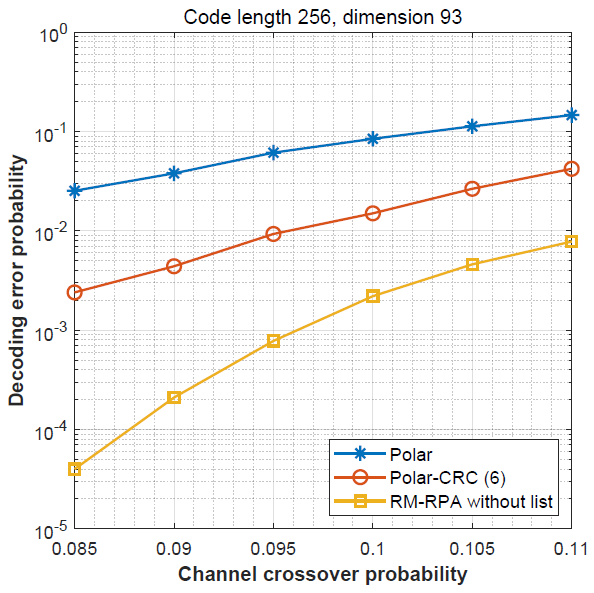}
\caption{$\cR\cM(8,3)$ vs Polar codes}
\end{subfigure}
~\hspace*{0.1in}
\begin{subfigure}{0.3\linewidth}
\centering
\includegraphics[width=\linewidth]{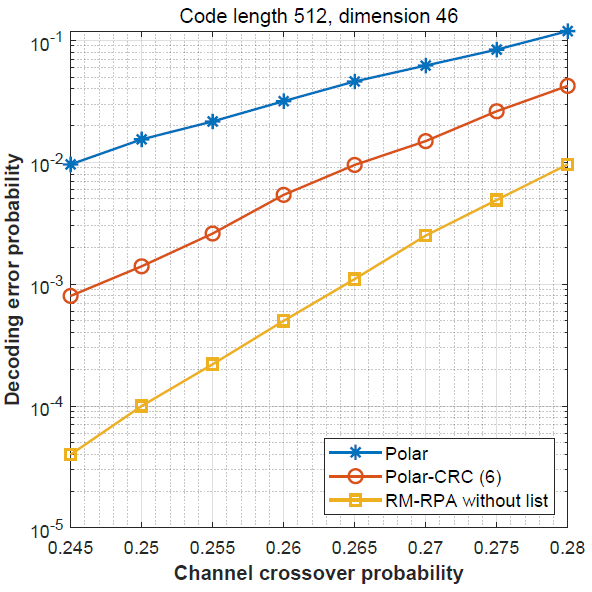}
\caption{$\cR\cM(9,2)$ vs Polar codes}
\end{subfigure}

\caption{Comparison between Reed-Muller codes and polar codes over BSC channels. For RM codes we use the RPA decoder in Algorithm~\ref{alg:highlvl} {\bf without list decoding}. For polar codes, no matter with or without CRC, we always use SCL decoder with list size $32$.}
\label{fig:cp}
\end{figure}

\begin{table}[h]
\begin{center}
\begin{tabular}{| c | c | c | c | c | c |}
\hline
$\cR\cM(7,2)$ & $P(7,2)$ & $\cR\cM(7,3)$ & $P(7,3)$ & $\cR\cM(7,4)$ & $P(7,4)$ \\ \hline
1ms & 7ms & 26ms & 15ms & 6ms & 23ms  \\  \hline 
\end{tabular}

\vspace*{0.1in}
\begin{tabular}{| c | c | c | c | c | c |}
\hline 
$\cR\cM(8,2)$ & $P(8,2)$ & $\cR\cM(8,3)$ & $P(8,3)$ & $\cR\cM(8,4)$ & $P(8,4)$ \\ \hline
4.3ms & 17ms & 236ms & 40ms & 5.9s & 64ms \\  \hline
\end{tabular}

\vspace*{0.1in}
\begin{tabular}{| c | c | c | c | c | c |}
\hline 
$\cR\cM(8,5)$ & $P(8,5)$ & $\cR\cM(9,2)$ & $P(9,2)$ & $\cR\cM(10,2)$ & $P(10,2)$ \\ \hline
14ms & 82ms & 18.2ms & 41ms & 76.7ms & 95ms \\ \hline
\end{tabular}
\caption{Comparison of decoding time between RM codes and polar codes. 
$P(m,r)$ denotes polar codes with the same length and dimension as $\cR\cM(m,r)$.}
\label{tb:rt}
\end{center}
\end{table}

\subsection{Comparison with previous decoding algorithms of RM codes}  \label{sect:pv}

We first compare with the decoding algorithm proposed by Sidel'nikov and Pershakov \cite{Sidel92}, which was later improved/modified in \cite{Loidreau04,Sakkour05}.
When decoding the second-order RM codes, the RPA decoding algorithm has some high-level similarity with the decoding algorithms in \cite{Sidel92,Loidreau04,Sakkour05} in the sense that the first step in all these algorithms is to project the received word $y$ onto the cosets of all the $n-1$ one-dimensional subspaces and decode the projected first-order RM codewords to obtain $\hat{y}_{/\mathbb{B}_1}, \hat{y}_{/\mathbb{B}_2}, \dots, \hat{y}_{/\mathbb{B}_{n-1}}$.
However, the next steps in \cite{Sidel92,Loidreau04,Sakkour05} are quite different from the RPA decoding algorithm and result in a worse performance than the RPA algorithm. More precisely, the main differences are:
\begin{itemize}
    \item The decoding algorithms in \cite{Loidreau04,Sakkour05} only work for the second order RM codes. For higher-order RM codes, the decoding algorithm proposed in \cite{Sidel92} is completely different from the RPA algorithm, and their performance is much worse than the RPA algorithm; see Fig.~\ref{fig:jet}(c).
    \item For second order RM codes, after the projection step, the RPA algorithm make use of both the decoding results of the projected codewords $\hat{y}_{/\mathbb{B}_1}, \hat{y}_{/\mathbb{B}_2}, \dots, \hat{y}_{/\mathbb{B}_{n-1}}$ and the original received word $y$ to obtain the final decoding results while the algorithms in \cite{Sidel92,Loidreau04,Sakkour05} only make use of $\hat{y}_{/\mathbb{B}_1}, \hat{y}_{/\mathbb{B}_2}, \dots, \hat{y}_{/\mathbb{B}_{n-1}}$ to obtain the coefficients of all the degree-2 monomials\footnote{Recall Definition~\ref{def:kee} and the discussion following it.} in the final decoding results. As discussed above, the projected codewords are more noisy than the original received words $y$. As a consequence, the performance of the algorithms in \cite{Sidel92,Loidreau04,Sakkour05} is worse than that of the RPA algorithm; see Fig.~\ref{fig:jet}(a),(b).
    \item The RPA algorithm uses $\hat{y}_{/\mathbb{B}_1}, \hat{y}_{/\mathbb{B}_2}, \dots, \hat{y}_{/\mathbb{B}_{n-1}}$ together with the original received word $y$ to correct errors bitwise in the original received word $y$ while the algorithms in \cite{Sidel92,Loidreau04,Sakkour05} use $\hat{y}_{/\mathbb{B}_1}, \hat{y}_{/\mathbb{B}_2}, \dots, \hat{y}_{/\mathbb{B}_{n-1}}$ to correct errors wordwise among themselves.
\end{itemize}

In Fig.~\ref{fig:jet}, we compare the RPA algorithm with the algorithms in \cite{Sidel92,Loidreau04,Sakkour05} for decoding Reed-Muller codes over AWGN and BSC channels. Note that there are two parameters $s$ and $h$ in the Sidelnikov-Pershakov algorithm, where $s$ is the list size of decoding each projected codeword, and $h$ is the number of iterations when decoding the projected codewords. In our simulations, we set $s=4$ and $h=3$ since larger values of $s$ and $h$ will not further improve the performance.

Next we compare the RPA algorithm with Dumer's recursive list decoding algorithm \cite{Dumer04,Dumer06,Dumer06a}. Dumer's list decoding algorithm provides a tradeoff between the decoding error probability and the decoding time. More precisely, if we set the list size to be large enough (e.g., exponential in $n$), then we can achieve the same performance as the maximal likelihood decoder, but we will also need exponential running time. If we choose small list size, then the algorithm runs fast but the decoding error will deteriorate.

In our simulations, we use the RPA algorithm and Dumer's algorithm to decode RM codes over AWGN channels, and we find that the decoding error probability of RPA is slightly better (smaller) than Dumer's algorithm, but the running time of RPA is typically larger. We have tested two cases $\cR\cM(8,2)$ and $\cR\cM(9,3)$, and the performance is given in Fig.~\ref{fig:Dumer}. For $\cR\cM(8,2)$, the running time of our algorithm is 4.3ms, and the running time of Dumer's algorithm is 0.85ms. For $\cR\cM(9,3)$, the running time of our algorithm is 3s, and the running time of Dumer's algorithm is 0.14s.



In \cite{Santi18}, simulation results are presented for $\cR\cM(7,3)$.  Their results are based on applying belief propagation to all minimum weight parity checks.  This does seem indirectly related to using all first-order RM subcodes to decode.  For $\cR\cM(7,3)$, the decoding complexities of these two approaches are also similar.  For RPA, each of $127*63$ projections takes roughly $32*5$ operations to decode, giving 1.2M operations per iteration.  For the algorithm in \cite{Santi18}, there are 94448 minimum weight parity checks of weight 16 giving roughly 1.5M operations per iteration. It turns out that for $\cR\cM(7,3)$, both the performance and the running time of RPA decoder are similar to the algorithm in \cite{Santi18}.

We also note that in \cite{Hashemi18}, an algorithm with near-ML performance was also provided for $\cR\cM(7,3)$.

\begin{figure}
\centering
\begin{subfigure}{0.3\linewidth} 
\centering
\includegraphics[width=\linewidth]{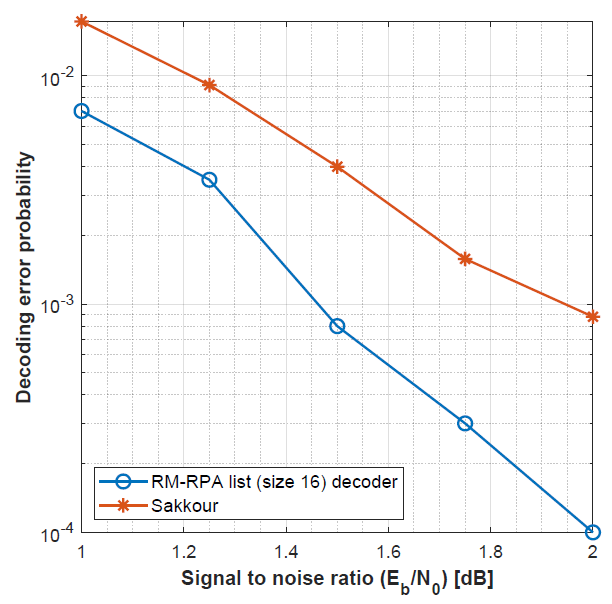}
\caption{$\cR\cM(9,2)$ over AWGN}
\end{subfigure}
~\hspace*{0.1in}
\begin{subfigure}{0.3\linewidth} 
\centering
\includegraphics[width=\linewidth]{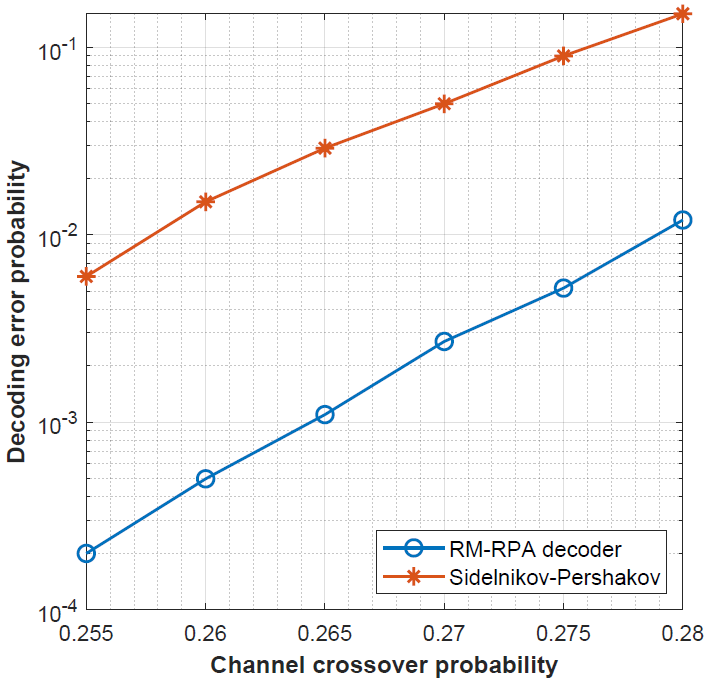}
\caption{$\cR\cM(9,2)$ over BSC}
\end{subfigure}
~\hspace*{0.1in}
\begin{subfigure}{0.3\linewidth}
\centering
\includegraphics[width=\linewidth]{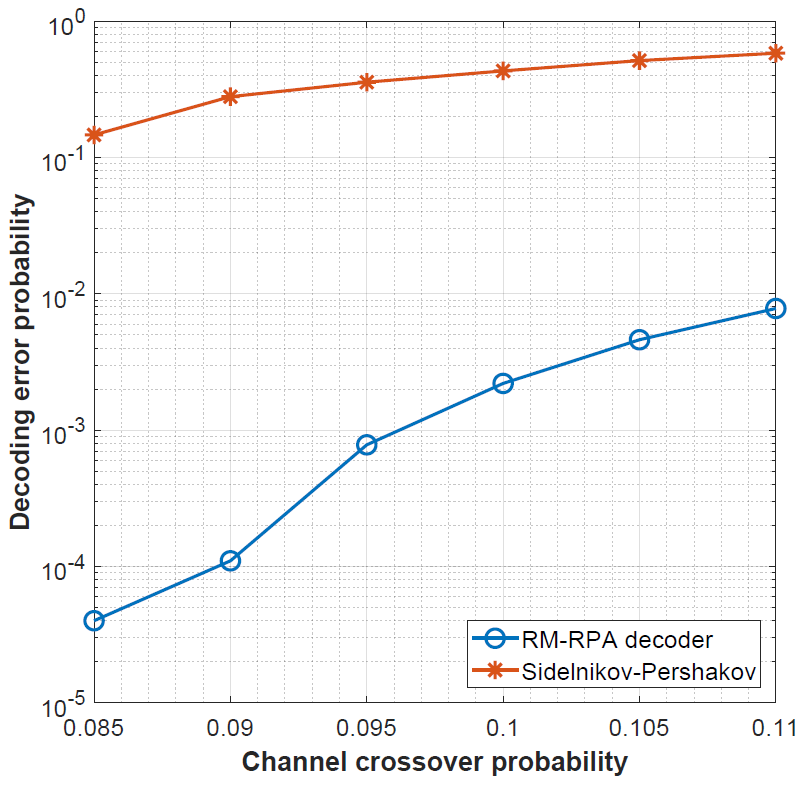}
\caption{$\cR\cM(8,3)$ over BSC}
\end{subfigure}
\caption{Comparison between the RPA algorithm and the algorithms in \cite{Sidel92,Loidreau04,Sakkour05} for decoding Reed-Muller codes over AWGN and BSC channels. The curve with legend ``Sakkour" is the performance of the algorithm in \cite{Loidreau04,Sakkour05}, and the curves with legend ``Sidelnikov-Pershakov" represent the performance of the algorithms in \cite{Sidel92}.}
\label{fig:jet}
\end{figure}

\begin{figure}
\centering
\begin{subfigure}{0.45\linewidth} 
\centering
\includegraphics[width=\linewidth]{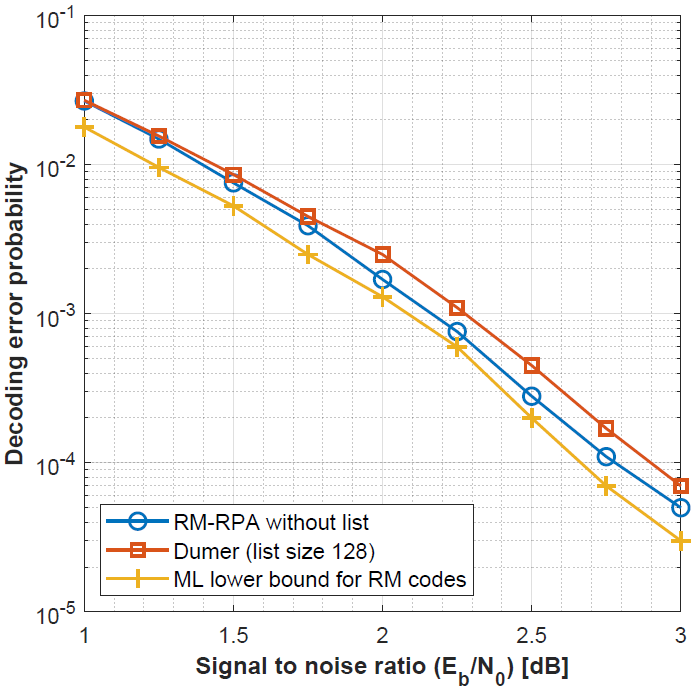}
\caption{$\cR\cM(8,2)$}
\end{subfigure}
~\hspace*{0.2in}
\begin{subfigure}{0.45\linewidth}
\centering
\includegraphics[width=\linewidth]{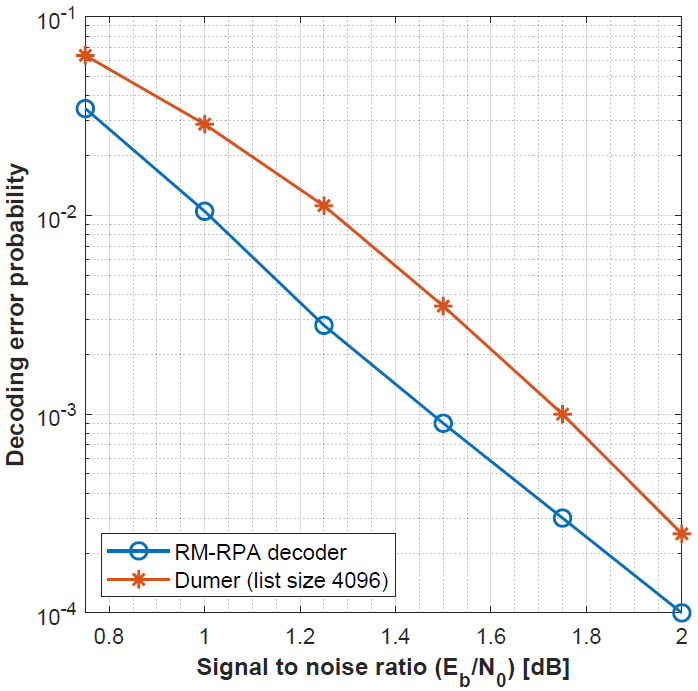}
\caption{$\cR\cM(9,3)$}
\end{subfigure}
\caption{Comparison between the RPA decoding algorithm {\bf without list} and Dumer's recursive list decoding algorithm (the algorithm described in Section III of \cite{Dumer06a}) for decoding Reed-Muller codes over AWGN channels.}
\label{fig:Dumer}
\end{figure}

\subsection{Parallelization and acceleration}  \label{sect:PandA}
Another important advantage of the new decoding algorithm for RM codes over the SCL decoder for polar codes is that our algorithm naturally allows parallel implementation while the SCL decoder is not parallelizable.
The key step in our algorithm for decoding a codeword of RM$(r,m)$ is to decode the quotient space codes which are in RM$(r-1,m-1)$ codes, and each of these can be decoded in parallel. 
Such a parallel structure is crucial to achieving high throughput and low latency.

Another way to accelerate the algorithm is to use only certain ``voting sets":
In the projection step, we can take a subset of one-dimensional subspaces instead of all the one-dimensional subspaces. Then we still use recursive decoding followed by the aggregation step. In this way, we decode fewer RM$(r-1,m-1)$ codes, and if the voting sets were chosen properly, we would obtain a similar decoding error probability with shorter running time.
Note that in Section~\ref{sect:highhigh} we already gave a concrete choice of voting set in Algorithm~\ref{alg:simp}, which indeed accelerates the decoding of high-rate RM codes with nearly-ML decoding error probability. At the same time, there might be other good voting sets to explore.

\subsection{Comparison with the meta converse bound for optimal codes \cite{PPV10,Polyanskiy10}}

We compared with upper bound from Corollary 39 and lower bound from Theorem 40 in \cite{Polyanskiy10}. More precisely, we provide the target error probability, the noise parameter of the channel, and the code dimension, then Corollary 39 and  Theorem 40 in \cite{Polyanskiy10} give us upper and lower bound on the (optimal) code length.
We found that $\cR\cM(8,2)$ is nearly optimal in terms of code length in the sense that the lower bound of code length given by \cite[Theorem 40]{Polyanskiy10} is 251, which differs from the actual code length of RM codes by only 5. Then $\cR\cM(9,2)$ is also close to optimal, where the lower bound on code length is 500. However, for RM codes with larger order (dimension) and larger code length, the lower bound differs from the actual code length by at least $50$, e.g., for $\cR\cM(9,3)$, the lower bound becomes 464.

\subsection{Optimal scaling and sharp threshold of Reed-Muller codes over BSC channels}
Recently, Hassani et al. gave theoretical results backing the conjecture that RM codes have an almost optimal scaling-law over BSC channels under ML decoding \cite{Hassani18}, where optimal scaling-law means that for a fixed linear code, the decoding error probability of ML decoder transitions from $0$ to $1$ as a function of the crossover probability of the BSC channel in the sharpest manner (i.e., comparable to random codes).
In particular, this implies that RM codes have sharper transition than polar codes under ML decoding (if capacity achieving).
In this section we give simulation results that show that for BSC channels, Reed-Muller codes under the RPA decoder also have sharper transition than polar codes under SCL+list decoder.

In Figure~\ref{fig:trans}, we plot the decoding error probability of RM codes and polar codes over BSC channels as a function of the channel crossover probability, where for RM codes we use the RPA decoder in Algorithm~\ref{alg:highlvl}, and for polar codes we use SCL decoder with list size $32$.
We can see that in all $4$ cases, the transition in the curve of RM codes is sharper than the transition in the curve of polar codes.
To further quantify the transition width, we introduce the following common notation:
Let us denote the channel crossover probability as $\epsilon$. For a given code and a corresponding decoding algorithm, we write its decoding error probability over BSC$(\epsilon)$ as $P_e(\epsilon)$.
For $0<\delta<1/2$, we define the transition width\footnote{Typically $P_e(\epsilon)$ is an increasing function of $\epsilon$, so the inverse function exists.}
$$
w(\delta):= P_e^{-1}(1-\delta)  -  P_e^{-1}(\delta).
$$
Clearly, $w(\delta)$ is a decreasing function. For a fixed value of $\delta$, smaller $w(\delta)$ means sharper transition and better scaling of the code and the corresponding decoder.

In Figure~\ref{fig:width}, we compare $w(0.1)$ and $w(0.01)$ between RM codes and polar codes with the same parameters, where we use the same decoders as above. We can see that RM codes always have smaller transition width than polar codes.
Moreover, within the same code family, the transition width $w(0.1)$ and $w(0.01)$ both decrease with the code length, meaning that the transition becomes sharper as the code length increases. This phenomena has already been proved for ML decoders in \cite{Tillich00} and \cite{Hassani18}.

\begin{figure}
\centering
\begin{subfigure}{0.47\linewidth} 
\centering
\includegraphics[width=\linewidth]{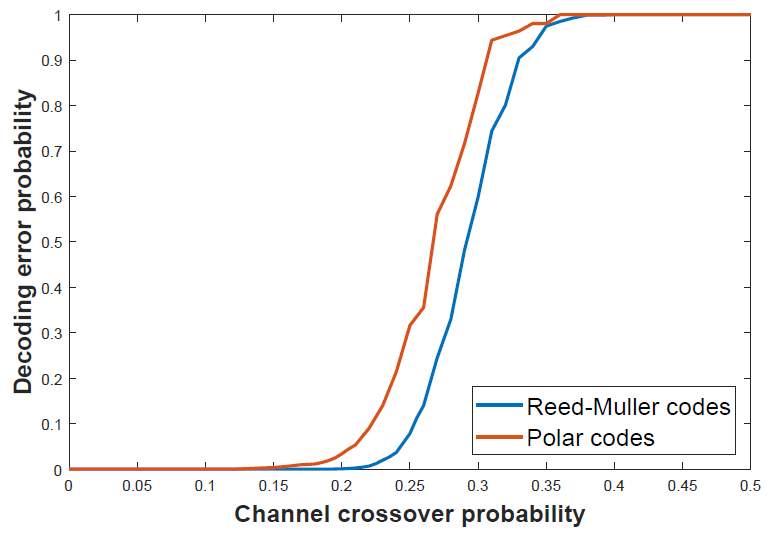}
\caption{$\cR\cM(8,2)$ vs Polar codes with the same parameters}
\end{subfigure}
~\hspace*{0.1in}
\begin{subfigure}{0.47\linewidth} 
\centering
\includegraphics[width=\linewidth]{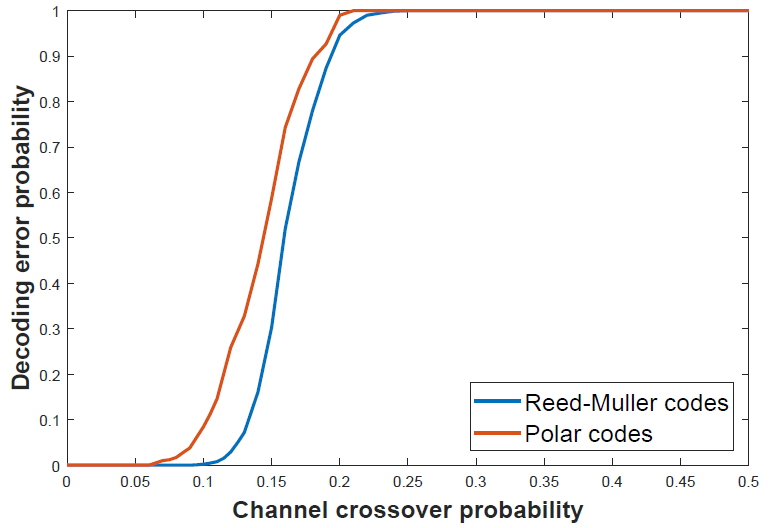}
\caption{$\cR\cM(8,3)$ vs Polar codes with the same parameters}
\end{subfigure}

\vspace*{0.1in}
\begin{subfigure}{0.47\linewidth} 
\centering
\includegraphics[width=\linewidth]{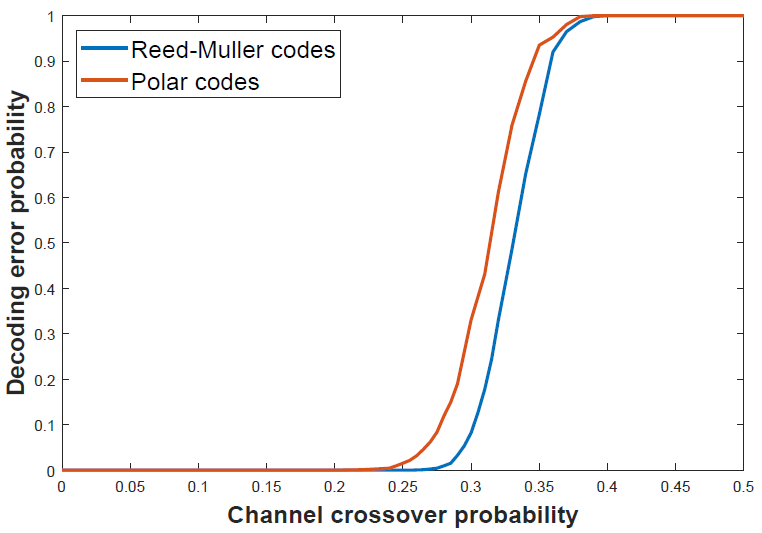}
\caption{$\cR\cM(9,2)$ vs Polar codes with the same parameters}
\end{subfigure}
~\hspace*{0.1in}
\begin{subfigure}{0.47\linewidth} 
\centering
\includegraphics[width=\linewidth]{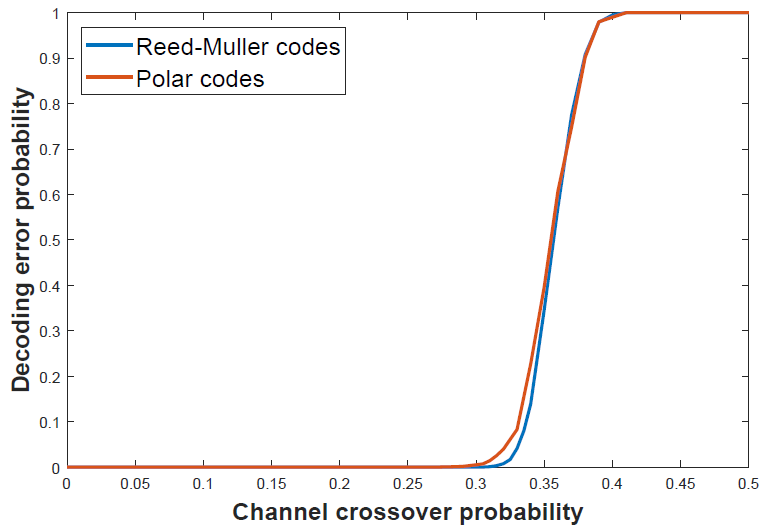}
\caption{$\cR\cM(10,2)$ vs Polar codes with the same parameters}
\end{subfigure}
\caption{Decoding error probability over BSC channels as a function of the channel crossover probability}
\label{fig:trans}
\end{figure}

\begin{figure}
\centering
\begin{subfigure}{0.47\linewidth} 
\centering
\includegraphics[width=\linewidth]{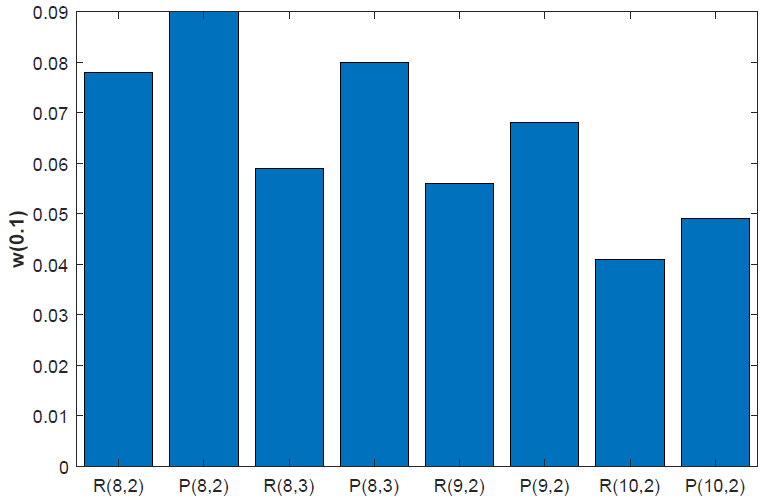}
\end{subfigure}
~\hspace*{0.1in}
\begin{subfigure}{0.47\linewidth} 
\centering
\includegraphics[width=\linewidth]{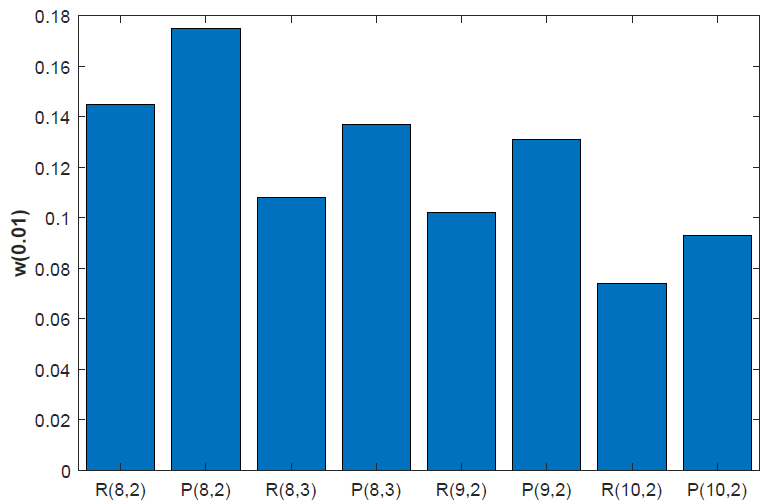}
\end{subfigure}
\caption{Comparison of transition width $w(0.1)$ and $w(0.01)$ between different codes. $R(m,r)$ refers to Reed-Muller codes, and $P(m,r)$ refers to polar codes with the same length and dimension as $R(m,r)$.}
\label{fig:width}
\end{figure}

\section{Extensions}
Here we mention a few possible extensions of the decoding algorithms.

1. The ``voting sets" idea to further accelerate the RPA decoding, as employed in Section~\ref{sect:highhigh} and discussed in Section~\ref{sect:PandA}.

2. Our new algorithms make use of one-dimensional subspace reduction. In practice, we can change the $\mathbb{B}_1,\dots,\mathbb{B}_{n-1}$ in the RPA decoding algorithms to any of the $s$-dimensional subspaces, with different combinations possible.
Note that in Section~\ref{sect:highhigh}, we already made use of this idea, where we chose $s=2$.

3. The RPA decoding algorithms can also be used to decode other codes that are supported on a vector space, or any code that has a well-defined notion of ``code projection'' that can be iteratively applied to produce eventually a trivial code (that can be decoded efficiently). In the case of RM codes, the quotient space projection has the specificity of producing again RM codes, and the trivial code is the Hadamard code that can be decoded using the FHT.

4. As discussed in Section~\ref{sect:spec}, we can use spectral decompositions or other relaxations in the Aggregation step instead of the majority voting, and depending on the regimes, one may take multiple iteration of the power-iteration method.

\section*{Acknowledgment} 
We thank Alexander Barg and Ilya Dumer for pointing out several references and giving useful feedback. We also thank Kirill Ivanov for useful discussions and feedback.

\bibliographystyle{IEEEtran}
\bibliography{decodeRM}

\appendices
\section{Proof of Lemma~\ref{lm:tbs}}  \label{ap:tbs}

Let $\mathbi{b}_1,\mathbi{b}_2,\dots,\mathbi{b}_m$ be a basis of $\mathbb{E}$ over $\mathbb{F}_2$ such that the first $s$ vectors $\mathbi{b}_1,\mathbi{b}_2,\dots,\mathbi{b}_s$ form a basis of $\mathbb{B}$.
Let $\mathbi{e}_1,\mathbi{e}_2,\dots,\mathbi{e}_m$ be the standard basis of $\mathbb{E}$, i.e., all but the $i$-th coordinate of $\mathbi{e}_i$ are $0$.
Then there is an $m\times m$ invertible matrix $M$ such that 
$$
(\mathbi{b}_1,\mathbi{b}_2,\dots,\mathbi{b}_m)^T=M (\mathbi{e}_1,\mathbi{e}_2,\dots,\mathbi{e}_m)^T.
$$
Let $(z_1,z_2,\dots,z_m)$ be the coordinates of a point in $\mathbb{E}$ under the standard basis
$(\mathbi{e}_1,\mathbi{e}_2,\dots,\mathbi{e}_m)$, and let $(z_1',z_2',\dots,z_m')$ be the coordinates of the same point under the basis $(\mathbi{b}_1,\mathbi{b}_2,\dots,\mathbi{b}_m)$. Then
$$
(z_1',z_2',\dots,z_m')=(z_1,z_2,\dots,z_m) M^{-1}.
$$
Notice that $\mathbb{B}=\{\mathbi{z}:(z_1',z_2',\dots,z_s')\in\mathbb{F}_2^s,z_{s+1}'=z_{s+2}'=\dots=z_m'=0\}$. Therefore for every coset $T\in\mathbb{E}/\mathbb{B}$, the last $m-s$ coordinates under the basis $(\mathbi{b}_1,\mathbi{b}_2,\dots,\mathbi{b}_m)$ are the same for all the points in $T$. As a result, we can use binary vectors of length $m-s$ to label the cosets, i.e.,
$$
[a_1,a_2,\dots,a_{m-s}]:=\{\mathbi{z}:(z_1',z_2',\dots,z_s')\in\mathbb{F}_2^s,z_{s+1}'=a_1,z_{s+2}'=a_2,\dots,z_m'=a_{m-s}\}.
$$
Next we associate every subset $A\subseteq [m]$ with another row vector $\mathbi{v}_m'(A)$ of length $2^m$, whose components are indexed by $\mathbi{z}=(z_1,z_2,\dots,z_m) \in \mathbb{E}$.
The vector $\mathbi{v}_m'(A)$ is defined as follows:
$$
\mathbi{v}_m'(A,\mathbi{z}) = \prod_{i\in A} z_i',
$$
where $\mathbi{v}_m'(A,\mathbi{z})$ is the component of $\mathbi{v}_m'(A)$ indexed by $\mathbi{z}$,
i.e., $\mathbi{v}_m'(A,\mathbi{z})$ is the evaluation of the polynomial $\prod_{i\in A}Z_i'$ at $\mathbi{z}$, where $(Z_1',Z_2',\dots,Z_m')=(Z_1,Z_2,\dots,Z_m) M^{-1}$.
Since all the invertible linear transforms belong to the automorphism group of Reed-Muller codes \cite{Macwilliams77}, we have the following alternative characterization of RM codes
$$
\cR\cM(m,r) := \left\{\sum_{A\subseteq[m],|A|\le r}u'(A) \mathbi{v}_m'(A): u'(A)\in\{0,1\} 
\text{~~for all~} A\subseteq[m],|A|\le r\right\}.
$$

It is easy to check that
for every coset $T=[z_{s+1}',z_{s+2}',\dots,z_m']\in \mathbb{E}/\mathbb{B}$, if $[s]\subseteq A$ then $\sum_{\mathbi{z}\in T} \mathbi{v}_m'(A,\mathbi{z}) = \prod_{i\in (A\setminus[s])} z_i'$, and if $[s]\nsubseteq A$ then $\sum_{\mathbi{z}\in T} \mathbi{v}_m'(A,\mathbi{z}) =0$.
Now let $c$ be a codeword of $\cR\cM(m,r)$, then it can be written as $c=\sum_{A\subseteq[m],|A|\le r}u'(A) \mathbi{v}_m'(A)$,
and for every coset $T=[z_{s+1}',z_{s+2}',\dots,z_m']\in \mathbb{E}/\mathbb{B}$, we have 
$$
\sum_{\mathbi{z}\in T} c(\mathbi{z})  =
\sum_{A\supseteq[s],|A|\le r}u'(A) \prod_{i\in (A\setminus[s])} z_i'         
 = \sum_{A\subseteq([m]\setminus[s]) ,|A|\le r-s}u'(A) \prod_{i\in A} z_i' .
$$
Therefore every codeword in $\cQ(m,r,\mathbb{B})$ corresponds to an $(m-s)$-variate polynomial in $\mathbb{F}_2[Z_{s+1}',Z_{s+2}',\dots,Z_m']$ with degree at most $r-s$, and this is exactly the definition of the $(r-s)$-th order Reed-Muller code $\cR\cM(m-s,r-s)$.

\section{Proof of Proposition~\ref{prop:eqde}}  \label{ap:eqde}
We need the following technical lemma to prove Proposition~\ref{prop:eqde}.
\begin{lemma}  \label{lm:ll}
Let $c_0=(c_0(\mathbi{z}), \mathbi{z}\in\mathbb{E})$ be a codeword of $\cR\cM(m,r)$. Let $L^{(1)}=(L^{(1)}(\mathbi{z}), \mathbi{z}\in\mathbb{E})$ and $L^{(2)}=(L^{(2)}(\mathbi{z}), \mathbi{z}\in\mathbb{E})$ be two LLR vectors such that 
\begin{equation} \label{eq:cd1}
L^{(2)}(\mathbi{z})=(-1)^{c_0(\mathbi{z})}L^{(1)}(\mathbi{z})  \quad\quad \forall  \mathbi{z}\in\mathbb{E}.
\end{equation}
Denote $\hat{c}_1=\texttt{RPA\_RM}(L^{(1)},m,r,N_{\max},\theta)$ and $\hat{c}_2=\texttt{RPA\_RM}(L^{(2)},m,r,N_{\max},\theta)$.
Then $\hat{c}_1=\hat{c}_2+c_0$.
\end{lemma}
\begin{proof}
We prove by induction on $r$.
For the base case $r=1$, we use the ML decoder as described at the beginning of this section. More precisely, according to \eqref{eq:mxr}, $\hat{c}_2=\texttt{RPA\_RM}(L^{(2)},m,1,N_{\max},\theta)$ is the codeword in $\cR\cM(m,1)$ that maximizes
$$
\sum_{\mathbi{z}\in\mathbb{E}} \Big( (-1)^{c(\mathbi{z})} L^{(2)}(\mathbi{z}) \Big),
$$
i.e.,
$$
\sum_{\mathbi{z}\in\mathbb{E}} \Big( (-1)^{\hat{c}_2(\mathbi{z})} L^{(2)}(\mathbi{z}) \Big)
\ge \sum_{\mathbi{z}\in\mathbb{E}} \Big( (-1)^{c(\mathbi{z})} L^{(2)}(\mathbi{z}) \Big)
\quad\quad \forall c\in\cR\cM(m,1).
$$
By \eqref{eq:cd1}, we have
$$
\sum_{\mathbi{z}\in\mathbb{E}} \Big( (-1)^{\hat{c}_2(\mathbi{z}) \oplus c_0(\mathbi{z})} L^{(1)}(\mathbi{z}) \Big)
\ge \sum_{\mathbi{z}\in\mathbb{E}} \Big( (-1)^{c(\mathbi{z}) \oplus c_0(\mathbi{z})} L^{(1)}(\mathbi{z}) \Big)
\quad\quad \forall c\in\cR\cM(m,1).
$$
Since $c_0$ is a codeword of $\cR\cM(m,1)$, we have:
$c_0+\cR\cM(m,1)=\cR\cM(m,1)$. As a result,
$$
\sum_{\mathbi{z}\in\mathbb{E}} \Big( (-1)^{\hat{c}_2(\mathbi{z}) \oplus c_0(\mathbi{z})} L^{(1)}(\mathbi{z}) \Big)
\ge \sum_{\mathbi{z}\in\mathbb{E}} \Big( (-1)^{c(\mathbi{z})} L^{(1)}(\mathbi{z}) \Big)
\quad\quad \forall c\in\cR\cM(m,1).
$$
Therefore, $\hat{c}_2\oplus c_0$ is the codeword in $\cR\cM(m,1)$ that maximizes
$$
\sum_{\mathbi{z}\in\mathbb{E}} \Big( (-1)^{c(\mathbi{z})} L^{(1)}(\mathbi{z}) \Big).
$$
Thus we conclude that $\hat{c}_1=\hat{c}_2\oplus c_0$. This establishes the base case.

For the inductive step, let us assume that the lemma holds for $r-1$ and prove it for $r$. Notice that in Algorithm~\ref{alg:genhighlvl}, $\hat{c}(\mathbi{z})$ is simply determined by the sign of $L(\mathbi{z})$. It is easy to see that if in Algorithm~\ref{alg:genAg}, the updated LLR vectors $\hat{L}^{(1)}$ and $\hat{L}^{(2)}$ always satisfy \eqref{eq:cd1}, then $\hat{c}_1=\hat{c}_2\oplus c_0$.
Therefore, we only need to prove \eqref{eq:cd1} for the updated LLR vectors $\hat{L}^{(1)}$ and $\hat{L}^{(2)}$.

Assuming that $L^{(1)}$ and $L^{(2)}$ satisfy \eqref{eq:cd1}, our task is to show that 
$\hat{L}^{(2)}(\mathbi{z})=(-1)^{c_0(\mathbi{z})} \hat{L}^{(1)}(\mathbi{z})$ for all $\mathbi{z}\in\mathbb{E}$.
From the analysis in Section~\ref{sect:gen}, we know that 
\begin{equation} \label{eq:yz}
\hat{L}^{(i)}(\mathbi{z})=\frac{1}{n-1}\sum_{\mathbi{z}'\neq \mathbi{z}} \alpha_i(\mathbi{z},\mathbi{z}')L^{(i)}(\mathbi{z}')
\text{~~for~} i=1,2.
\end{equation}
The coefficient $\alpha_i(\mathbi{z},\mathbi{z}')$ is $1$ if the decoding result of the corresponding $(r-1)$th order RM code at the coset $\{\mathbi{z},\mathbi{z}'\}$ is $0$, and $\alpha_i(\mathbi{z},\mathbi{z}')$ is $-1$ if the decoding result at the coset $\{\mathbi{z},\mathbi{z}'\}$ is $1$ (see line 3 of Algorithm~\ref{alg:genAg}).

Next we will show that $\alpha_2(\mathbi{z},\mathbi{z}')=(-1)^{c_0(\mathbi{z})\oplus c_0(\mathbi{z}')} \alpha_1(\mathbi{z},\mathbi{z}')$.
Note that $\alpha_i(\mathbi{z},\mathbi{z}')$ is determined by the decoding result
$\hat{y}_{/\mathbb{B}}^{(i)} = \texttt{RPA\_RM}(L_{/\mathbb{B}}^{(i)},m-1,r-1,N_{\max}, \theta)$, where $\mathbb{B}=\{0,\mathbi{z}\oplus \mathbi{z}'\}$.
By \eqref{eq:lt}, we have
\begin{align*}
L_{/\mathbb{B}}^{(2)}(T) & =\ln \Big( \exp \big( \sum_{\mathbi{z}\in T} L^{(2)}(\mathbi{z}) \big) +1 \Big) - 
\ln \Big( \sum_{\mathbi{z}\in T} \exp(L^{(2)}(\mathbi{z})) \Big) \\
& =\ln \Big( \exp \big( \sum_{\mathbi{z}\in T} (-1)^{c_0(\mathbi{z})} L^{(1)}(\mathbi{z}) \big) +1 \Big) - \ln \Big( \sum_{\mathbi{z}\in T} \exp \big( (-1)^{c_0(\mathbi{z})} L^{(1)}(\mathbi{z}) \big) \Big) \\
& = (-1)^{\bigoplus_{\mathbi{z}\in T} c_0(\mathbi{z})} 
\left( \ln \Big( \exp \big( \sum_{\mathbi{z}\in T} L^{(1)}(\mathbi{z}) \big) +1 \Big) - 
\ln \Big( \sum_{\mathbi{z}\in T} \exp(L^{(1)}(\mathbi{z})) \Big) \right) \\
& = (-1)^{\bigoplus_{\mathbi{z}\in T} c_0(\mathbi{z})} L_{/\mathbb{B}}^{(1)}(T).
\end{align*}
Let us write $c_0(T):=\bigoplus_{\mathbi{z}\in T} c_0(\mathbi{z})$.
Then $L_{/\mathbb{B}}^{(2)}(T)=(-1)^{c_0(T)}L_{/\mathbb{B}}^{(1)}(T)$ for all $T\in\mathbb{E}/\mathbb{B}$. Moreover, since $c_0$ is a codeword of $\cR\cM(m,r)$ and $\mathbb{B}$ is a one-dimensional subspace of $\mathbb{E}$, by Lemma~\ref{lm:tbs} we know that $(c_0(T),T\in\mathbb{E}/\mathbb{B})$ is a codeword of $\cR\cM(m-1,r-1)$. Therefore, the codeword $(c_0(T),T\in\mathbb{E}/\mathbb{B})$ and the two LLR vectors $(L_{/\mathbb{B}}^{(1)}(T),T\in\mathbb{E}/\mathbb{B})$ and $(L_{/\mathbb{B}}^{(2)}(T),T\in\mathbb{E}/\mathbb{B})$ satisfy the conditions of this lemma. By the induction hypothesis,
$\hat{y}_{/\mathbb{B}}^{(2)}(T) = \hat{y}_{/\mathbb{B}}^{(1)}(T) \oplus c_0(T)$ for all $T\in\mathbb{E}/\mathbb{B}$.
As a result, we have $\alpha_2(\mathbi{z},\mathbi{z}')=(-1)^{c_0(\mathbi{z})\oplus c_0(\mathbi{z}')} \alpha_1(\mathbi{z},\mathbi{z}')$.
Taking this into \eqref{eq:yz}, we conclude that for all $\mathbi{z}\in\mathbb{E}$,
\begin{align*}
\hat{L}^{(2)}(\mathbi{z}) & =\frac{1}{n-1}\sum_{\mathbi{z}'\neq \mathbi{z}} \alpha_2(\mathbi{z},\mathbi{z}')L^{(2)}(\mathbi{z}') \\
& =\frac{1}{n-1}\sum_{\mathbi{z}'\neq \mathbi{z}} \Big( (-1)^{c_0(\mathbi{z})\oplus c_0(\mathbi{z}')} \alpha_1(\mathbi{z},\mathbi{z}')
(-1)^{c_0(\mathbi{z}')}L^{(1)}(\mathbi{z}') \Big) \\
& = (-1)^{c_0(\mathbi{z})} \frac{1}{n-1} \sum_{\mathbi{z}'\neq \mathbi{z}}  \alpha_1(\mathbi{z},\mathbi{z}')    L^{(1)}(\mathbi{z}') = (-1)^{c_0(\mathbi{z})} \hat{L}^{(1)}(\mathbi{z}).
\end{align*}
This completes the proof of the inductive step and establishes the lemma.

\end{proof}

\underline{\em Proof of Proposition~\ref{prop:eqde}:}
Since $W$ is a BMS channel, there is a permutation $\pi$ of the output alphabet $\cW$ satisfying the two conditions in Definition~\ref{def:bms}. Since both $c_1$ and $c_2$ are codewords of $\cR\cM(m,r)$, $c_0:=c_1+c_2$ is also a codeword of $\cR\cM(m,r)$.
Clearly, both channel output vectors $Y_1$ and $Y_2$ belong to $\cW^n$.
Now we define a permutation $\pi^{c_0}$ on $\cW^n$: For any $y=(y(\mathbi{z}),\mathbi{z}\in\mathbb{E}) \in \cW^n$, 
$$
\pi^{c_0}(y) := (\pi^{c_0(\mathbi{z})}(y(\mathbi{z})), \mathbi{z}\in\mathbb{E}).
$$
Notice that $c_0(\mathbi{z})$ is either $0$ or $1$, and $\pi^0$ is the identity map. Since $\pi$ is a permutation on $\cW$, $\pi^{c_0}$ is clearly a permutation on $\cW^n$.
For a given $y=(y(\mathbi{z}),\mathbi{z}\in\mathbb{E}) \in \cW^n$, we denote the LLR vector corresponding to $y$ as $L_{y}^{(1)}:=(L_{y}^{(1)}(\mathbi{z}),\mathbi{z}\in\mathbb{E})$, i.e.,
$L_{y}^{(1)}(\mathbi{z})=\LLR(y(\mathbi{z}))$ for all $\mathbi{z}\in\mathbb{E}$, and we denote the LLR vector corresponding to $\pi^{c_0}(y)$ as $L_y^{(2)}:=(L_y^{(2)}(\mathbi{z}),\mathbi{z}\in\mathbb{E})$, i.e.,
$L_y^{(2)}(\mathbi{z})=\LLR(\pi^{c_0(\mathbi{z})}(y(\mathbi{z})))$ for all $\mathbi{z}\in\mathbb{E}$.
By the property of $\pi$ (see Definition~\ref{def:bms}), we have
$$
L_y^{(2)}(\mathbi{z})=(-1)^{c_0(\mathbi{z})}L_y^{(1)}(\mathbi{z})  \quad\quad \forall  \mathbi{z}\in\mathbb{E}.
$$
Since $c_0\in\cR\cM(m,r)$, by Lemma~\ref{lm:ll} we know that
$$
\texttt{RPA\_RM}(L_y^{(1)},m,r,N_{\max},\theta) = \texttt{RPA\_RM}(L_y^{(2)},m,r,N_{\max},\theta) +c_0.
$$
As a result, $\texttt{RPA\_RM}(L_y^{(1)},m,r,N_{\max},\theta)\neq c_1$ if and only if $\texttt{RPA\_RM}(L_y^{(2)},m,r,N_{\max},\theta) \neq c_2$.

For a vector $y\in\cW^n$ and a codeword $c\in\cR\cM(m,r)$, we use $W^n(y|c)$ to denote the probability of outputting $y$ when the transmitted codeword is $c$. Again by the property of $\pi$, it is easy to see that
$$
W^n(y|c_1) = W^n(\pi^{c_0}(y) | c_2)            \quad\quad \forall  y\in\cW^n.
$$
Recall that in Proposition~\ref{prop:eqde}, we use $L^{(1)}$ and $L^{(2)}$ to denote the random LLR vectors corresponding to the random channel outputs when transmitting $c_1$ and $c_2$, respectively.
Therefore,
\begin{align*}
& \mathbb{P}(\texttt{RPA\_RM}(L^{(1)},m,r,N_{\max},\theta)\neq c_1) \\
= & \sum_{y\in\cW^n} W^n(y|c_1) \mathbbm{1}[\texttt{RPA\_RM}(L_y^{(1)},m,r,N_{\max},\theta)\neq c_1] \\
= & \sum_{y\in\cW^n} W^n(\pi^{c_0}(y) | c_2) \mathbbm{1}[\texttt{RPA\_RM}(L_y^{(2)},m,r,N_{\max},\theta) \neq c_2]  \\
= & \mathbb{P}(\texttt{RPA\_RM}(L^{(2)},m,r,N_{\max},\theta)\neq c_2).
\end{align*}
This completes the proof of Proposition~\ref{prop:eqde}.
 \hfill \qed

\clearpage

\section{Another version of Algorithm~\ref{alg:highlvl}--\ref{alg:BSCAg}}  \label{ap:newd}
\begin{algorithm}
\caption{The \texttt{RPA\_RM} decoding function for BSC}    \label{alg:newd}
\textbf{Input:} The corrupted codeword $y=(y(\mathbi{z}), \mathbi{z}\in \{0,1\}^m)$; the parameters of the Reed-Muller code $m$ and $r$; the maximal number of iterations $N_{\max}$

\textbf{Output:} The decoded codeword $\hat{c}$

\vspace*{0.05in}
\begin{algorithmic}[1]
\For {$i=1,2,\dots,N_{\max}$} 
\State Initialize $(\texttt{changevote}(\mathbi{z}), \mathbi{z}\in \{0,1\}^m)$ as an all-zero vector indexed by $\mathbi{z}\in \{0,1\}^m$
\For {each non-zero $\mathbi{z}_0 \in \{0,1\}^m$}
\State Set $\mathbb{B}=\{0,\mathbi{z}_0\}$
\State $\hat{y}_{/\mathbb{B}} \gets \texttt{RPA\_RM}(y_{/\mathbb{B}},m-1,r-1,N_{\max})$
\State \Comment{If $r=2$, then we use the Fast Hadamard Transform to decode the first-order RM code \cite{Macwilliams77}}
\For {each $\mathbi{z}\in \{0,1\}^m$}
\If {$y_{/\mathbb{B}}([\mathbi{z}+\mathbb{B}]) \neq \hat{y}_{/\mathbb{B}} ([\mathbi{z}+\mathbb{B}])$}
\State $\texttt{changevote}(\mathbi{z}) \gets \texttt{changevote}(\mathbi{z})+1$ \Comment{Here addition is between real numbers}
\EndIf
\EndFor
\EndFor
\State $\texttt{numofchange} \gets 0$
\State $n \gets 2^m$
\For {each $\mathbi{z}\in  \{0,1\}^m$}
\If {$\texttt{changevote}(\mathbi{z})>\frac{n-1}{2}$}
\State $y(\mathbi{z}) \gets y(\mathbi{z}) \oplus 1$ \Comment{Here addition is over $\mathbb{F}_2$}
\State $\texttt{numofchange} \gets \texttt{numofchange}+1$   \Comment{Here addition is between real numbers}
\EndIf
\EndFor
\If {$\texttt{numofchange} = 0$}
\State \textbf{break}  \Comment{Exit the first for loop of this function}
\EndIf
\EndFor
\State $\hat{c} \gets y$
\State \textbf{return} $\hat{c}$
\end{algorithmic}
\end{algorithm}

\clearpage

\section{Another version of Algorithms~\ref{alg:genhighlvl}--\ref{alg:genAg}}  \label{ap:general}

\begin{algorithm}
\caption{The \texttt{RPA\_RM} decoding function for general binary-input memoryless channels}    \label{alg:general}
\textbf{Input:} The LLR vector $(L(\mathbi{z}), \mathbi{z}\in \{0,1\}^m)$; the parameters of the Reed-Muller code $m$ and $r$; the maximal number of iterations $N_{\max}$; the exiting threshold $\theta$

\textbf{Output:} The decoded codeword $\hat{c}= (\hat{c}(\mathbi{z}), \mathbi{z}\in \{0,1\}^m)$

\vspace*{0.05in}
\begin{algorithmic}[1]
\State  $\mathbb{E} := \{0,1\}^m$
\For {$i=1,2,\dots,N_{\max}$} 
\State Initialize $(\texttt{cumuLLR}(\mathbi{z}), \mathbi{z}\in\mathbb{E})$ as an all-zero vector indexed by $\mathbi{z}\in\mathbb{E}$
\For {each non-zero $\mathbi{z}_0 \in\mathbb{E}$}
\State Set $\mathbb{B}=\{0,\mathbi{z}_0\}$
\State $L_{/\mathbb{B}} \gets (L_{/\mathbb{B}}(T),T\in\mathbb{E}/\mathbb{B})$
\Comment{$L_{/\mathbb{B}}(T)$ is calculated from $(L(\mathbi{z}), \mathbi{z}\in\mathbb{E})$ according to \eqref{eq:lt}}
\State $\hat{y}_{/\mathbb{B}} \gets \texttt{RPA\_RM}(L_{/\mathbb{B}},m-1,r-1,N_{\max}, \theta)$
\State \Comment{If $r=2$, then we use the Fast Hadamard Transform to decode the first-order RM code}
\For {each $\mathbi{z}\in\mathbb{E}$}
\If {$\hat{y}_{/\mathbb{B}}([\mathbi{z}+\mathbb{B}])=0$}
\State $\texttt{cumuLLR}(\mathbi{z}) \gets \texttt{cumuLLR}(\mathbi{z})+L(\mathbi{z}\oplus \mathbi{z}_0)$ 
\Else  \Comment{$\hat{y}_{/\mathbb{B}}$ is the decoded codeword, so $\hat{y}_{/\mathbb{B}}([\mathbi{z}+\mathbb{B}])$ is either $0$ or $1$}
\State $\texttt{cumuLLR}(\mathbi{z}) \gets \texttt{cumuLLR}(\mathbi{z})-L(\mathbi{z}\oplus \mathbi{z}_0)$ 
\EndIf
\EndFor
\EndFor
\State $\texttt{numofchange} \gets 0$
\State $n \gets 2^m$
\For {each $\mathbi{z}\in\mathbb{E}$}
\State $\texttt{cumuLLR}(\mathbi{z}) \gets \frac{\texttt{cumuLLR}(\mathbi{z})} {n-1}$ 
\If {$|\texttt{cumuLLR}(\mathbi{z})-L(\mathbi{z})|>\theta |L(\mathbi{z})|$}
\State $\texttt{numofchange} \gets \texttt{numofchange}+1$   \Comment{Here addition is between real numbers}
\EndIf
\State $L(\mathbi{z}) \gets \texttt{cumuLLR}(\mathbi{z})$
\EndFor
\If {$\texttt{numofchange} = 0$}
\State \textbf{break}  \Comment{Exit the first for loop of this function}
\EndIf
\EndFor
\For {each $\mathbi{z}\in\mathbb{E}$}
\If {$L(\mathbi{z}) > 0$}
\State $\hat{c}(\mathbi{z}) \gets 0$
\Else  
\State $\hat{c}(\mathbi{z}) \gets 1$
\EndIf
\EndFor
\State \textbf{return} $\hat{c}$
\end{algorithmic}
\end{algorithm}

\end{document}